\documentclass[10pt,journal,a4paper]{IEEEtran}
\usepackage{amssymb,amsmath}
\usepackage{cite}
\usepackage{psfrag}
\usepackage{url}
\usepackage[latin1]{inputenc}
\usepackage[absolute,overlay]{textpos}

\usepackage{pgf}
\usepackage[latin1]{inputenc}
\usepackage[export]{adjustbox}
\usepackage{hhline}
\usepackage{multirow}
\usepackage{arydshln}
\usepackage{bm}
\usepackage{bbm}
\usepackage{booktabs}
\usepackage{dsfont}
\usepackage{dblfloatfix}
\usepackage{upgreek}
\usepackage{relsize}
\usepackage{eufrak}
\DeclareMathOperator{\Tr}{Tr}

\usepackage[linesnumbered,ruled]{algorithm2e}

\usepackage{amsthm}

\newtheorem{proposition}{Proposition}

\newtheorem{theorem}{Theorem}
\newtheorem{corollary}{Corollary}
\newtheorem{remark}{Remark}

\begin{document}
	
	\title{CSI-free vs CSI-based multi-antenna WET for massive low-power Internet of Things}
	\author{
		\IEEEauthorblockN{Onel L. A. L\'opez, \emph{Member, IEEE}, 
			Nurul Huda Mahmood, \emph{Member, IEEE}, 
			Hirley Alves, \emph{Member, IEEE}, and 
			Matti Latva-aho, \emph{Senior Member, IEEE}
		}
		\thanks{Authors are affiliated to the Centre for Wireless Communications (CWC), University of Oulu, Finland. \{onel.alcarazlopez,nurulhuda.mahmood, hirley.alves,matti.latva-aho\}@oulu.fi}	
		\thanks{This work is supported by Academy of Finland (Aka) (Grants n.307492, n.318927 (6Genesis Flagship), n.319008 (EE-IoT)).}}
	
	\maketitle

	\begin{abstract}
		Wireless Energy Transfer (WET) is a promising solution for powering massive Internet of Things 
		deployments. An important question is whether the costly Channel State Information (CSI) acquisition procedure is necessary for optimum performance. In this paper, we shed some light into this matter by evaluating   CSI-based and CSI-free multi-antenna WET schemes in a setup with WET in the downlink, and periodic or Poisson-traffic Wireless Information Transfer (WIT) in the uplink. When CSI is available, we show that a maximum ratio transmission beamformer is close to optimum whenever the farthest node experiences at least 3 dB of power attenuation more than the remaining devices. On the other hand, although the adopted CSI-free mechanism  is not capable of providing average harvesting gains, it does provide greater WET/WIT  diversity with lower energy requirements when compared with the CSI-based scheme. Our numerical results evidence that the CSI-free scheme performs favorably under periodic traffic conditions, but it may be deficient in case of Poisson traffic, specially if the setup is not optimally configured.  Finally, we show the prominent performance results when the uplink transmissions are periodic, while highlighting the need of a minimum mean square error equalizer rather than zero-forcing for information decoding.
	\end{abstract}
	\section{Introduction}\label{introduction}
	The Internet of Things (IoT) is a major technology trend that promises to interconnect
	\textit{everything} towards building a data-driven society enabled by near-instant unlimited wireless connectivity \cite{Matti.2019,Mahmood.2020}.
	A key feature/challenge of the IoT is the massive connectivity since around 80 billion connected devices are foreseen to proliferate globally by 2025, thus resulting in a massive technology-led disruption across all industries \cite{sullivan.2020}.
	
	The IoT ranges from cloud (e.g., data centers, super computers, internet core network) and fog (e.g., computers, smartphones, smart appliances) technologies, to edge (e.g., wearables, smart sensors, motes) and extreme edge (e.g., smart dust and zero-power sensors) technologies \cite{Portilla.2019}. Energy efficiency and/or power consumption criteria become more critical as one descends over such layers. In fact, edge or extreme edge devices are usually powered by batteries or energy harvesters and are very limited in computing and storage capabilities to reduce costs and enlarge lifetime. 
	Many types of energy harvesting (EH) technologies are under consideration, but those relying on wireless  radio frequency (RF) signals are becoming more and more attractive.  
	RF-EH provides key benefits such as \cite{Niyato.2017,Lopez.2019}: i) battery charging without physical connections, which significantly simplify the servicing and maintenance of battery-powered devices; ii) readily available service in the form of transmitted energy (TV/radio broadcasters, mobile base stations and handheld radios), iii) low cost and form factor reduction of the end devices; iv) increase of durability and reliability of end devices thanks to their contact-free design; and v) enhanced energy efficiency and network-wide reduction of emissions footprint.
	
	RF-EH is a wide concept\footnote{Herein we focus on RF-EH networks where the RF signals are intentionally transmitted for powering the EH devices. Alternatively, the devices may 		opportunistically harvest energy from RF signals of different frequencies already in their surrounding environment and to which they are sensitive. The latter is known as  ambient RF EH, and readers can refer to \cite{Ghazanfari.2016} for an overview.} that encompasses two main scenarios when combined with Wireless Information Transfer (WIT), namely  Wireless Powered Communication Network (WPCN) and Simultaneous Wireless Information and Power Transfer (SWIPT) \cite{Lopez.2019}. In the first scenario,  a Wireless Energy Transfer (WET) process occurs in the downlink in a first phase and WIT takes place in the second phase. Meanwhile, in the second scenario, WET and WIT occur simultaneously.
	An overview of the recent advances on both architectures can be found in \cite{Clerckx.2019}, while herein the discussions will focus on WPCN and pure WET setups. Notice that WET may have a much more significant role than WIT in practical applications as highlighted in \cite{Lopez.2019}. This is because WET's duration is often required to be the largest i) in order to harvest usable amounts of energy, and/or ii) due to sporadic WIT rounds, e.g., event-driven traffic.
	Since SWIPT may happen just occasionally, WPCN use cases are often of much more practical interest. 
	Therefore, enabling efficient WPCNs is mandatory  \cite{Mahmood.2020,Lopez.2019,Mahmood.2019}, and constitutes the scope of this work.
	\subsection{Related Work}\label{RW}
	Over the past few years, the analysis and optimization of WPCNs has evolved  from the simple Harvest-then-Transmit (HTT) protocol \cite{Ju.2014,Lopez.2017,Huang.2016}  towards more evolved alternatives that are capable of boosting the system performance either via cooperation \cite{Chen.2015}, power control \cite{Lopez.2018}, rate allocation schemes \cite{LopezFernandez.2018} and/or  retransmissions \cite{Makki.2016}. However, most of the works so far are concerned with rather optimistic setups where either i) most of the power consumption sources at the EH devices are ignored, ii) Channel State Information (CSI) procedures are assumed cost free, and iii) only one or few EH devices are powered. Regarding the latter, the number of EH devices is often not greater than the number of powering antennas such that full gain from energy beamforming (EB) is attained in the WET phase, e.g., \cite{Huang.2016,Cantos.2019,Du.2020}. For instance,  a setup where a multi-antenna hybrid access point (HAP) transfers power to the devices via EB, followed by the devices sending their data simultaneously by consuming the harvested energy, is investigated in \cite{Cantos.2019}. The authors cast a max-min rate optimization problem with practical non-linear EH and solve it via several iterative optimization methods. However, no other power consumption sources besides transmissions are considered, and Zero Forcing (ZF) equalization is used for information decoding at the HAP  without analysing the CSI acquisition costs. Meanwhile, the authors in \cite{Du.2020} do consider the CSI acquisition costs when optimizing the HAP pilots power and the power allocated to the energy transmission, while the EH devices are under the effect of several power consumption sources. Yet, the imposition of having more antennas than devices may be strong towards  future low-power massive IoT networks. Finally, the lack of a traffic source model for data transmissions is also a strong limitation for most of the works, which intrinsically assume full-buffer EH devices, e.g., \cite{Ju.2014,Lopez.2017,Huang.2016,Chen.2015,Lopez.2018,LopezFernandez.2018,Makki.2016,Cantos.2019,Du.2020}.
	
	One important observation is that the gains from EB decrease quickly as the number of EH IoT devices increases \cite{Lopez.2019}. This holds even without accounting for the considerable energy resources demanded by CSI acquisition. Therefore, in massive deployment scenarios, the broadcast nature of wireless transmissions should be intelligently exploited for powering simultaneously a massive number of IoT devices with minimum or no CSI \cite{Lopez.2019,LopezMahmood.2020}. To that end, the authors in \cite{Clerckx.2018} propose a new form of signal design for WET relying on phase sweeping transmit diversity, which forces the multiple antennas to induce fast fluctuations of the wireless channel and does not rely on any form of CSI. This is accomplished by exploiting the non-linearity of the EH circuitry, however, the attained diversity gain is quite small even when the transmitter is equipped with massive antenna arrays. Meanwhile, several multi-antenna CSI-free WET solutions have been recently proposed and analyzed in \cite{Lopez.2019_CSI,Lopez.2020} to improve the statistics of the RF energy availability at the input of the EH circuitry of a massive set of energy harvesters:
	\begin{itemize}
		\item One Antenna ($\mathrm{OA}$), under which the power beacon (PB) transmits with only one antenna; 
		\item  All Antennas transmitting the Same Signal ($\mathrm{AA-SS}$), under which the PB transmits the same signal simultaneously with all antennas but with reduced power at each; 
		\item All Antennas transmitting Independent Signals ($\mathrm{AA-IS}$), under which the PB  transmits power signals independently generated across the antennas; and 
		\item Switching Antenna ($\mathrm{SA}$), under which the PB transmits with full power by one antenna at a time such that all antennas are used during a coherence block\footnote{All antennas need to be used (at least once but never concurrently) in a coherence block, which can be easily guaranteed without specific CSI in static or semi-static setups (as typical in WPCNs), where fading is sufficiently slow.}. 
	\end{itemize}
	Notice that i) $\mathrm{OA}$ is the simplest scheme since it does not take advantage of the multiple spatial resources, while ii) $\mathrm{AA-SS}$ may reach considerable gains in terms of average harvested energy under Line of Sight (LOS) but it is highly sensitive to the different mean phases of the LOS channel component,  and iii) $\mathrm{AA-IS}$, $\mathrm{SA}$ do not improve the average energy availability but do provide transmit diversity. It was demonstrated in \cite{Lopez.2020} that devices closer to the PB benefit more from $\mathrm{AA-IS}$, while those that are far, and more likely to operate near their sensitivity level, benefit more from the $\mathrm{SA}$. 
	All these CSI-free WET schemes have been considered without the information communication component typical of a WPCN, and consequently, their influence on the overall system performance is so far unclear. 
	\subsection{Contributions and Organization of the Paper}\label{contributions}	
	This paper aims at analyzing for the first time the gains from operating with/without CSI for powering massive low-power IoT deployments with uplink transmission requirements. Specifically, we consider a WPCN where a massive set of IoT nodes require occasional uplink information transmissions to a HAP, which in turn is constantly transferring RF energy to them in the downlink.
	Herein, we adopt the $\mathrm{SA}$ strategy \cite{Lopez.2019_CSI,Lopez.2020} as the CSI-free WET scheme, which, besides the benefits aforementioned, allows a better coupling to the co-located information transmission processes. The latter is because   only one antenna is used for WET at any time, while the remaining antennas stay silent, thus all these idle antennas may be used for uplink information decoding in WPCN setups. For information decoding in the uplink, the HAP implements either ZF or the Minimum Mean Square Error (MMSE) equalization. The latter is shown to provide large performance gains for the WPCN under consideration when compared to ZF, mainly because of the low-rate low-power transmissions, which are typical in the analyzed scenario.
	The
	main contributions of this work are listed as follows:
	\begin{itemize}
		\item We investigate and analyze a WPCN setup under CSI-based and CSI-free powering schemes. We are concerned with the overall outage probability, which encompasses both WET and WIT processes' failures. 
		The performance is evaluated in terms of the worst node's performance such that we can assure Quality of Service (QoS) guarantees for all nodes in the network. We consider the power consumption from several sources, e.g., transmission, circuitry, and CSI-acquisition procedures;
		\item We decouple WET and WIT processes and cast a max-min WET optimization problem when CSI is available at the HAP. 
		We provide analytical bounds on the performance of the CSI-based WET beamforming by relying on Maximum Ratio Transmission (MRT). We show that the MRT is near the fairest EB, e.g., the EB that provides max-min performance guarantees, even in a massive deployment, if the farthest EH node experiences at least 3 dB of power attenuation more than the remaining devices;
		\item We consider two types of information traffic sources: i) periodic traffic, such that the network is perfectly synchronized; and ii) Poisson traffic, which is uncoordinated and random. The overall performance is analyzed for both traffic profiles. Our results not only evidence that the system performance deteriorates under Poisson random access when compared to deterministic traffic, but also that it is more challenging to optimally configure the network. We cast an optimization problem to determine the optimum pilot reuse factor such that the collision probability keeps below a certain limit. A solution algorithm is provided and shown to converge in few iterations;
		\item The impact of the CSI-based and CSI-free scheme on the WET performance is analytically analyzed and several trade-offs are identified. It is shown that the CSI-free scheme is preferable as the number of IoT devices increases and/or the CSI acquisition costs increase. In terms of overall performance, the CSI-free scheme is shown to perform favorably under periodic traffic conditions, but it may be deficient in case of Poisson traffic, specially if the setup is not optimally configured.
	\end{itemize}
	
	Next, Section~\ref{system} presents the system model and assumptions, Section~\ref{wet} discusses the energy outage performance under the CSI-based and CSI-free WET schemes, while Section~\ref{WIT} addresses the information outage performance under ZF and MMSE decoding schemes. Section~\ref{results} presents and discusses numerical results. Finally,
	Section~\ref{conclusions} concludes the paper.
	
	\textbf{Notation:} Boldface lowercase letters denote column vectors, while boldface uppercase letters denote matrices. For instance, $\mathbf{x}=\{x_i\}$ where $x_i$ is the $i$-th element of vector $\mathbf{x}$; while $\mathbf{X}=\{X_{i,j}\}$ where $X_{i,j}$ is the $i$-th row $j$-th column element of matrix $\mathbf{X}$. By $\mathbf{I}$ we denote the identity matrix, and by $\mathbf{1}$ we denote a vector of ones. Superscripts $(\cdot)^T$ and $(\cdot)^H$ denote the transpose and conjugate transpose operations, while $\Tr(\cdot)$ and $\mathrm{diag}(\mathbf{x})$ denote the trace operator and a diagonal matrix with elements $\{x_i\}$, respectively. $\mathbb{C}$, $\mathbb{R}$ and $\mathbb{Z}^+$ are the set of complex, real and non-negative integer numbers, respectively; while $\bm{i}=\sqrt{-1}$ is the imaginary unit and $\Im(x)$ denotes the imaginary part of $x\in\mathbb{C}$. The absolute/cardinality operations in case of scalars/sets is denoted as $|\cdot|$, while $||\mathbf{x}||$ denotes the euclidean norm of vector $\mathbf{x}$.
	Additionally, $\lfloor\cdot\rfloor$ and	$\lceil\cdot\rceil$ are the floor and ceiling functions, respectively, while $\sup\{\cdot\}$ and $\inf\{\cdot\}$ are the supremum and infimum notations.
	The curled inequality symbol $\succeq$ is used  to indicate positive definiteness of a matrix, while $\mathcal{O}(\cdot)$ is the big-O notation.
	$\mathbb{E}_X[\!\ \cdot\ \!]$ denotes expectation with respect to random variable (RV) $X$, which is characterized by a Probability Density Function (PDF) $f_X(x)$ and Cumulative Distribution Function (CDF) $F_X(x)$, while $\mathbb{P}[A]$ is the probability of event $A$.
	Also, $\sum_Y X$ denotes the sum of $Y$ RVs distributed as $f_X(x)$.
	$\mathbf{c}\sim\mathcal{CN}(\bm{\mu},\mathbf{R})$ is a circularly-symmetric Gaussian complex random vector with mean $\bm{\mu}$ and covariance $\mathbf{R}$, while $Y\sim\chi^2(\varphi,\psi)$ is a non-central chi-squared RV with $\varphi$ degrees of freedom and parameter $\psi$ such that \cite{Kobayashi.2011}
	\begin{align}
	F_Y(y)&=1-\mathcal{Q}_{\varphi/2}\big(\sqrt{\psi},\sqrt{y}\big),\label{cdf}
	\end{align}
	where $\mathcal{Q}_a(\cdot)$ denotes the Marcum Q-function, which is given in \cite[Eq. (1)]{Nuttall.1975}.
	\section{System model}\label{system}
	We consider the scenario depicted in Fig.~\ref{Fig1}. In the downlink, a HAP  wirelessly  powers  a large set $\mathcal{S}=\{s_i\}$ of $\mathrm{S}$ single-antenna EH sensor nodes located nearby. Such low-power devices require in turn to sporadically send  some short data messages of $k$ bits/Hz over time blocks of $t$ seconds in the uplink. The HAP is equipped with $M$ antennas, $M_t$ of which  are used for downlink energy transmission, and the remaining $M_r=M-M_t$ for information decoding in the uplink.
	We assume that the coherence time $T_c$ is sufficiently large such that $t\le T_c/M$ for any feasible $M$. On the one hand, notice that since the RF-EH devices are extremely-low-power nodes, they are foreseen to be mostly static devices, thus, the coherence time is large. On the other hand, such devices are expected to transmit for short times due to intrinsically small data payloads, low-latency requirements, and/or lack of energy resources to support longer transmissions \cite{Khan.2016}. 
	Then, by limiting for instance the analysis of this work to $M\le M_0$ we can set $t=T_c/M_0$, although extending any of our analyses for any other smaller $t$ would be straightforward.
	\begin{figure}[t!]
		\centering  \includegraphics[width=0.7\columnwidth]{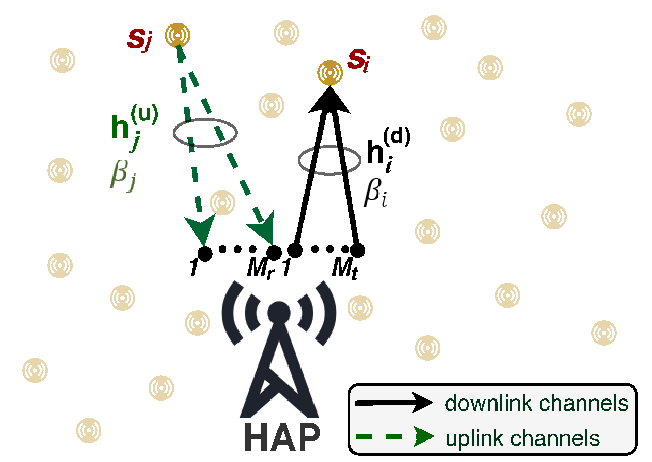}
		\caption{System model: a HAP equipped with $M$ antennas powers wirelessly in the downlink a set $\mathcal{S}$ of single-antenna sensor nodes located nearby, while it receives information from a subset of them in the uplink.}	
		\label{Fig1}
	\end{figure}
	\subsection{Channel model}
	The average channel gain between the HAP and $s_i$ is denoted as $\beta_i$, e.g. the path loss is $1/\beta_i$, while the small-scale fading channel coefficient between the HAP's antennas and $s_i$ (downlink) is denoted as  $\mathbf{h}^{(d)}_i\in\mathbb{C}^{M_t\times 1}$, and the channel between $s_i$ and HAP's antennas (uplink) is denoted as $\mathbf{h}^{(u)}_i\in\mathbb{C}^{M_r\times 1}$. Notice that even when the network is configured to operate over the same frequency band in uplink and downlink, the channel reciprocity is difficult to hold in this kind of setup since devices at both ends are extremely different \cite{Guillaud.2005}, hence we assume fully independent uplink and downlink channels\footnote{Even when certain dependence may exist, this does not affect significantly our results. This is because WIT phases in WPCNs are mostly sporadic, then, the aggregated harvested energy between consecutive WIT phases is much less dependent on the fading experienced in a particular coherence block.}. 
	
	The antenna elements are sufficiently separated such that the fading seen at each antenna can be assumed independent. We assume quasi-static channels undergoing Rician fading, i.e., $\mathbf{h}^{(d)}_i,\mathbf{h}^{(u)}_i\sim\mathcal{CN}\big(\sqrt{\frac{\kappa}{1+\kappa}}\bm{1}_{M_{\{t,r\}}\times 1},\frac{1}{1+\kappa}\mathbf{I}_{M_{\{t,r\}}\times M_{\{t,r\}}}\big)$, which is a very general assumption that allows  modeling a wide variety of channels by tuning the Rician factor $\kappa\ge 0$ \cite[Ch.2]{Proakis.2001}, e.g., when $\kappa=0$ the channel envelope is Rayleigh distributed, while when $\kappa\rightarrow\infty$ there is a fully deterministic LOS channel.
	\subsection{Transmission model}\label{txmodel}
	We assume homogeneous (in terms of hardware, supported services and traffic characterization) IoT devices which are harvesting RF energy from HAP's transmissions. They require $p_c$ power units to keep active, otherwise they are in outage. This value obviously depends on their circuitry but also on the services to support. Additionally, the EH devices need to report their data to the HAP at some moments, so they briefly interrupt  (during $t$ seconds) the EH to send it. We model such transmission activation in two different ways, by considering \cite{Nikaein.2013}:
	\begin{itemize}
		\item periodic traffic, such that the network is perfectly synchronized and every EH device has a predefined slot allocated for transmission. If the periodicity is $t_s$, then there are $\lfloor t_s/t\rfloor$ slots available. If $\mathrm{S}\le \lfloor t_s/t\rfloor$, then each device operates alone in the channel; otherwise there will be up to $\lceil \mathrm{S}/\lfloor t_s/t\rfloor\rceil$ concurrent transmissions eventually;
		\item Poisson traffic, such that the network traffic is uncoordinated. Let us take $\lambda$ as the mean number of messages per coherence time that are required to be transmitted by each device. Notice that it is evident that $\lambda<1$ needs to hold according to our previous discussions. 
	\end{itemize}
	It is worth noting that neither the periodic nor the Poisson model are suitable for mimicking bursty traffic, for which other more suitable models are recommended \cite{Nikaein.2013}. However, a WPCN implementation is not suitable in scenarios requiring bursty transmissions mostly due to its inherent and strict energy limitations, thus we resorted to the above simple but effective models covering two extreme ends. Additionally, note that the multiple antennas at the HAP require to be exploited for spatially separating  the concurrent transmissions with high reliability. We delve into the specific details in Section~\ref{WIT}.

	Finally, $p_i$ denotes the fixed transmit power of $s_i$, while $\xi_\mathrm{csi}^{(u)},\ \xi_\mathrm{csi}^{(d)}$ represent the energy resources\footnote{Notice that the time resource for $\xi_\mathrm{csi}^{(d)}$ is limited by $T_c/\mathrm{S}$, while for $\xi_\mathrm{csi}^{(u)}$ is limited by the overall transmission duration $t$. In fact,  we assume in our analyses that the uplink pilot training phase is much shorter than the actual data transmission and ignore its impact on the information outage performance in Section~\ref{WIT}.} (power $\times$ time) utilized by such EH node to let the HAP know the uplink and downlink CSI, respectively. 
    Notice that since channel reciprocity does not hold, it is expected that $\xi_\mathrm{csi}^{(u)}<\xi_\mathrm{csi}^{(d)}$ as transmissions from the EH devices are required in both downlink and uplink (pilot transmissions in uplink, feedback in downlink), but decoding/processing the pilots sent by the HAP is also required in the downlink. Note that although we do not model an imperfect CSI acquisition (perfect CSI is assumed whenever required)\footnote{However, some results evincing the degenerative effect of imperfect uplink CSI acquisition are discussed in Section~\ref{results}.}, e.g., due to estimation and quantization errors and processing delay, we do consider the associated  energy consumption costs. 
	\subsection{Performance evaluation}\label{performance}
	We adopt the outage probability formulation as the main performance metric. We say $s_i$ is in outage when: i) the harvested energy was insufficient for supporting its operation and consequently no uplink data transmission occurred: \textit{energy outage} $\mathcal{O}_i^{(d)}$, or ii) uplink transmission occurred but the transmitted message could not be decoded at the HAP: \textit{information outage} $\mathcal{O}_i^{(u)}$. Since downlink and uplink channels are independent and transmit powers are fixed we have that $s_i$'s outage probability is given by 
	\begin{align}
	\mathcal{O}_i&=1-\big(1-\mathcal{O}_i^{(d)}\big)\big(1-\mathcal{O}_i^{(u)}\big)\nonumber\\
	&=\mathcal{O}_i^{(d)}+\mathcal{O}_i^{(u)}-\mathcal{O}_i^{(d)}\mathcal{O}_i^{(u)}.\label{out}
	\end{align}
	Finally, the network performance is evaluated in terms of the worst node's performance by computing the network outage probability as
	\begin{align}
	\mathcal{O}=\sup_{i=1,\cdots,\mathrm{S}}\{O_i\}.
	\end{align}
	\begin{remark}
		Then, we can assure that every EH device in the network performs reliably at least the $(1-\mathcal{O})\%$ of time. Notice that for $\mathcal{O}_i\ll 1$, the term $\mathcal{O}_i^{(d)}+\mathcal{O}_i^{(u)}$  dominates \eqref{out}. Since this is required in practical scenarios, we can examine independently the bounds on $\mathcal{O}_i^{(d)}$ and $\mathcal{O}_i^{(u)}$.
	\end{remark}	
	\section{Wireless Energy Transfer}\label{wet}
	In Subsection~\ref{csi}, we first propose a CSI-based precoding scheme for optimizing the WET process. Then, we address the CSI-free WET alternative in Subsection~\ref{free}. The energy outage performance under both CSI-based and CSI-free WET schemes is also analyzed therein.
	\subsection{CSI-based WET}\label{csi}
	In each coherence block time, the HAP sends pilot signals that are used by the EH devices to estimate the downlink channels. Then, such information is fedback  to the HAP through the uplink channels in an ordered way. As commented before, in such processes, the EH devices spend $\xi_\mathrm{csi}^{(d)}$ energy units each time, which is approximately given  as 
	\begin{align}
	\xi_\mathrm{csi}^{(d)}\approx M_t\xi_0,\label{csid}
	\end{align}
	where $\xi_0$ denotes the energy required for decoding, processing and sending back to the HAP the information related to the pilot signals coming from each antenna. As we will show later in Subsection~\ref{users}, very often, the HAP only requires  the WET-CSI from a small set of EH devices, and therefore it is expected that their CSI feedback can be scheduled without overlapping.
	
	As there are $M_t$ transmit antennas, the HAP is able to transmit $M_t$ energy beams to broadcast energy to all sensors in $\mathcal{S}$. Then, the incident RF power at $s_i$ is given by
	\begin{align}\label{Ei}
	E_i^\mathrm{rf}&=\!\mathbb{E}_x\bigg[\!\Big(\!\sqrt{P\beta_i}(\mathbf{h}_i^{\!(d)})^T\!\sum_{j=1}^{M_t}\!\mathbf{w}_jx_j\Big)^{\!\!H}\!\Big(\!\sqrt{P\beta_i}(\mathbf{h}_i^{\!(d)})^T\!\sum_{j=1}^{M_t}\!\mathbf{w}_jx_j\!\Big)\!\bigg]\nonumber\\
	&=\!P\beta_i\sum_{j=1}^{M_t}\!\big|(\mathbf{h}_i^{(d)})^T\mathbf{w}_j\big|^2,
	\end{align}
	where $P$ is the HAP's transmit power, $\mathbf{w}_j\in\mathbb{C}^{M_t\times 1},\  j=1,\cdots,M_t$, denotes the precoding vector for generating the $j-$th energy beam, and $x_j$ is its normalized energy carrying signal, i.e., $\mathbb{E}[x_j^Hx_j]=1$, which is independently generated across the antennas, i.e., $\mathbb{E}[x_j^Hx_l]=0, \forall j\ne l$. 
	\subsubsection{Energy beamforming}
	For our setup and performance evaluation criterion, the optimum precoder $\{\mathbf{w}_j\}$ is the one that minimizes $\sup_{i=1,\cdots,\mathrm{S}}\ \{\mathcal{O}_i^{(d)}\}$. However, since the set $\{h_i^{(d)}\}$ is  known by the HAP after the CSI acquisition procedures, the problem translates to maximize $\inf_{i=1,\cdots,\mathrm{S}}\ \{E_i^\mathrm{rf}\}$ subject to $\sum_{j=1}^{M_t}||\mathbf{w}_j||^2\le 1$. The previous objective function is not concave and therefore the problem is not convex. However, it
	can still be optimally solved by rewriting it as a semi-definite programming (SDP) problem \cite{Thudugalage.2016} as given next.
	\begin{proposition}
		The optimum precoder $\{\mathbf{w}_j\}$ matches the normalized eigenvectors of $\mathbf{W}$, which is the solution of
		\begin{subequations}\label{P}
			\begin{alignat}{2}
			\mathbf{P:}\  &\underset{\mathbf{W}\in\mathbb{C}^{M_t\times M_t}, \zeta}{\mathrm{minimize}}       &\ \ & 
			-\zeta \label{P:a}\\ 
			&\text{subject to} & &  P\beta_i\Tr(\mathbf{W}\mathbf{H}_i^{(d)})\ge \zeta, \  i\!=\!1,\cdots,\mathrm{S} \label{P:b}\\ 
			& & &\qquad \Tr(\mathbf{W})= 1 \label{P:c}\\ 
			& & &\qquad\qquad \mathbf{W}\succeq 0, \label{P12:d}
			\end{alignat}	
		\end{subequations}
	    where  $\mathbf{H}_i^{(d)}=\mathbf{h}_i^{(d)}\mathbf{h}_i^{(d)H}$ and $\zeta$ is an auxiliary optimization variable.
	\end{proposition}
	\begin{proof}
		Let us define $\zeta \triangleq \inf_{i} \{E_i^\mathrm{rf}\}$, while $E_i^\mathrm{rf}$ in \eqref{Ei} can be rewritten as
		\begin{align}
		E_i^\mathrm{rf}=P\beta_i\sum_{j=1}^{M_t}\mathbf{h}_i^{(d)H}\mathbf{w}_j\mathbf{w}_j^H\mathbf{h}_i^{(d)}=P\beta_i\Tr(\mathbf{W}\mathbf{H}_i^{(d)}),\label{Ei2}
		\end{align}
		where $\mathbf{W}=\sum_{j=1}^{M_t}\mathbf{w}_j\mathbf{w}_j^H$ is a Hermitian matrix (with maximum rank $\min (\mathrm{S},M_t)$) that can be found by solving \eqref{P}. Notice that  \eqref{P:c} corresponds to the power budget constraint. Finally, the beamforming vectors $\{\mathbf{w}_j\}$ match the eigenvectors of $\mathbf{W}$ but normalized by their corresponding eigenvalues' square roots such that $\Tr(\mathbf{W})=1$.
	\end{proof}
	This procedure allows finding the optimum precoding vectors, and hereinafter it is referred to as CSI-based beamforming.
	\begin{remark}	\label{re2}	
	Interior point methods are mostly adopted to efficiently solve SDP problems. Since $\mathbf{P}$ consists of a linear function, $\mathrm{S}+1$ linear
	constraints, one positive semi-definite constraint, and the more challenging optimization variable has size $M_t \times M_t$, interior point methods will take $\mathcal{O}(M_t \log(1/\epsilon))$ iterations, with each iteration requiring at most $\mathcal{O}(M_t^6 + (\mathrm{S} + 1)M_t^2)$ arithmetic operations \cite{Ye.2011}, where $\epsilon$ represents the solution accuracy
	at the algorithm's termination. In addition, an eigendecomposition of $\mathbf{W}$, which has complexity $\mathcal{O}(M_t^3)$, is required in order to derive the set of beamforming vectors. Consequently, the
	SDP solution becomes computationally costly as the number of HAP's transmit antennas and/or the number of EH devices increases.
    \end{remark}
	\subsubsection{Energy outage lower bound}
	Notice that
	\begin{align}
	\sup_i\{\mathcal{O}_i^{(d)}\}&\ge \inf_{\{\mathbf{w}_j\},\forall j}\{\mathcal{O}_{i'}^{(d)}\},\label{O1}
	\end{align}
	where $s_{i'}$ is the sensor under the greatest path loss: $\beta_{i'}\!\le\! \beta_i,\ \forall s_i\in \mathcal{S}$, e.g., the farthest sensor. The above expression strictly holds  as long as we consider the same energy requirements for all devices, e.g., homogeneous devices with the same  transmit power $p_i=p,\ \forall s_i\in \mathcal{S}$. However, \eqref{O1} should also hold when intelligent power allocation polices are utilized.
	
	In the best possible scenario, where the HAP requires compensating only the channel impairments of $s_{i'}$ since the remaining nodes are under more favorable channel/propagation conditions, a MRT precoding will be the optimum. Such MRT precoding is indistinctly and equivalently given  by
	\begin{align}
	\mathbf{w}_j=\frac{1}{\sqrt{M_t}}\frac{\mathbf{h}_{i'}^{(d)*}}{||\mathbf{h}_{i'}^{(d)}||},\ \forall j,\ \ \mathrm{or}\
	\left\{\begin{array}{ll}
	\frac{\mathbf{h}_{i'}^{(d)*}}{||\mathbf{h}_{i'}^{(d)}||},& j=1\\
	\bm{0}, & j>1
	\end{array}\right.,\label{wj}
	\end{align} 
	for which $E_{i'}^\mathrm{rf}$ in each coherence interval becomes
	\begin{align}
	E_{i'}^\mathrm{rf}&=P\beta_{i'}||\mathbf{h}_{i'}^{(d)}||^2\sim \frac{P\beta_{i'}}{2(1+\kappa)}\mathfrak{X},\label{Ei'} 
	\end{align}
	which comes from using \cite[Eq.~(45)]{Lopez.2019_CSI} and setting $\mathfrak{X}\sim \chi^2(2M_t,2M_t\kappa)$.
	
	We assume that the energy harvested between consecutive uplink transmissions requires to be enough  for powering the circuits, performing the CSI acquisition procedures, and sending an uplink information message, while the remaining (if any) energy is used in other tasks, e.g., sensing, signal processing, etc. Therefore,
	\begin{itemize}
		\item for periodic traffic, the total energy harvested by $s_{i'}$ between its uplink transmissions is at most given by
		\begin{align}
		E_{i'}&\stackrel{(a)}{=}\eta T_c \sum_{\lceil t_s/t_c\rceil} E_{i'}^\mathrm{rf} \stackrel{(b)}{\sim}\frac{\eta T_c P\beta_i}{2(1+\kappa)} \sum_{{\lceil t_s/T_c\rceil}}\mathfrak{X}\nonumber\\
		&\stackrel{(c)}{\sim}\frac{\eta T_c P\beta_i}{2(1+\kappa)} \chi^2\big(2M_t  \lceil  t_s/T_c\rceil, 2M_t\kappa \lceil t_s/T_c\rceil \big),
		\end{align}
		where in $(a)$, $\eta\in(0,1)$ denotes the energy conversion efficiency\footnote{Note that we are considering a simple linear EH model as in \cite{Ghazanfari.2016,Ju.2014,Lopez.2017,Huang.2016,Chen.2015,LopezFernandez.2018,Makki.2016,Khan.2016,Thudugalage.2016} to allow some analytical tractability and facilitate the discussions. Although the specific performance results must vary when utilizing different EH models (as those in \cite{Lopez.2018,Lopez.2019_CSI,Clerckx.2019,Cantos.2019,Du.2020,Clerckx.2018,Lopez.2020}), the trends and relative performance gaps between the CSI-based and CSI-free schemes are expected to hold.}, and the summation is over $\lceil t_s/T_c\rceil$ independent RVs of the form of $E_{i'}^\mathrm{rf}$, $(b)$ comes from using \eqref{Ei'}, while $(c)$ follows after using the definition of a non-central chi-square RV. Notice that although we conveniently used  $\lceil t_s/t_c\rceil$ to take advantage of a finite summation, the last expression holds without such a constraint.
		
		Meanwhile, the energy requirements under periodic traffic are given by
		\begin{align}
		E_0&=\big\lceil t_s/T_c\big\rceil\xi_\mathrm{csi}^{(d)}+\xi_\mathrm{csi}^{(u)}+p_c t_s+p t,
		\end{align}
		thus \eqref{O1} becomes
		\begin{align}
		\sup_i&\{\mathcal{O}_i^{(d)}\}\nonumber\\
		\ge& \mathbb{P}\big[E_{i'}<E_0\big]\nonumber\\
		=& 1\!-\!\mathcal{Q}_{M_t\lceil \frac{t_s}{T_c} \rceil}\!\Bigg(\sqrt{2M_t\kappa\lceil\tfrac{t_s}{T_c}\rceil},\sqrt{\frac{2E_0(\kappa+1)}{\eta T_c P\beta_{i'}}}\Bigg),\label{per}
		\end{align}
		which comes from using the CDF of a non-central chi-square RV; while 
		\item  for Poisson traffic, the messages arrive with an exponential inter-arrival random time $U$ with mean $1/\lambda$ (given in coherence intervals). For analytical tractability, let us assume that transmissions also occur  in an slotted fashion, where slots are of duration $t$. Then, devices with a ready-to-send message wait for the next time slot for transmission. 
		
		Let us denote the inter-arrival RV  by $V$, which is now discrete and with PMF given by
		\begin{align}
		\mathbb{P}[V=v]&=\mathbb{P}[v-1\le U< v]\nonumber\\
		&=F_U(v)-F_U(v-1)\nonumber\\
		&=(e^{\lambda }-1)e^{-\lambda v}\label{Vpdf}
		\end{align}
		for $v\ge 1$. 
		Now, $E_{i'}$ becomes a random sum of $E_{i'}^\mathrm{rf}$ RVs, i.e.,
		\begin{align}
		E_{i'}=\eta T_c\sum\nolimits_V E_{i'}^\mathrm{rf},\label{Ei''}
		\end{align}
	    while the energy requirements to make the uplink transmissions take place are random as well, and can be written as
		\begin{align}
		E_0=&v\xi_\mathrm{csi}^{(d)}+\xi_\mathrm{csi}^{(u)}+p_c vT_c+p t.\label{E0}
		\end{align}
		Then, the energy outage lower bound is given next.
		\begin{theorem}\label{the2}
			Under Poisson traffic, \eqref{O1} becomes
			\begin{align}
			\sup_i\{\mathcal{O}_i^{(d)}\}&\gtrsim 1-(e^\lambda\!-\!1)\!\!\sum_{v=1}^{v_{\max}}\!e^{-\lambda v}\mathcal{Q}_{M_t v}\Big(\!\sqrt{\!2M_t\kappa v},\nonumber\\
			&\!\frac{2(1\!+\!\kappa)}{\eta T_c P\beta_{i'}}\!\sqrt{v\big(\xi_\mathrm{csi}^{(d)}\!+p_cT_c\big)\!+\!\xi_\mathrm{csi}^{(u)}\!+\!p t}\Big),\label{cfO}
			\end{align}
			where $v_{\max}$ can be chosen arbitrarily large such that $v_{\max}\ge 10\times \frac{e^{\lambda}}{e^\lambda-1}$.	
		\end{theorem}
	\begin{proof}
		See Appendix~\ref{App_A}.\phantom\qedhere
	\end{proof}
	\end{itemize}
	\subsubsection{On the optimality of the MRT beamforming}\label{users}
	Let us assume that the HAP is using the MRT beamformer to power the farthest node $s_{i'}$. One question arises: \textit{How such beamformer impacts the wireless powering of the remaining devices?} The following result sheds some light into this question.	
	%
	\begin{theorem}\label{the3}
		The larger
		\begin{align}
		\Omega=\frac{1}{\beta_{i'}}\bigg(\frac{1}{4}\Big(\frac{\kappa}{1\!+\!\kappa/\sqrt{2}}\Big)^2\!+\!\frac{1}{M_t(1\!+\!\kappa/2)}\bigg)\!\inf_{s_i\in\mathcal{S}\backslash s_{i'}}\!\{\beta_i\}\label{Ec}
		\end{align}
		is, the greater the chances of MRT being the optimum beamformer. In fact, when $\Omega>1$, the MRT beamformer is at least half of the time the optimum.		
	\end{theorem}
\begin{proof}
	See Appendix~\ref{App_B}.  \phantom\qedhere
\end{proof}
\begin{remark}\label{re3}
	Notice that if we consider the large-LOS scenario, \eqref{Ec} simplifies to 
	\begin{align}
	\Omega
	\stackrel{\kappa\rightarrow\infty}{\approx}\frac{1}{2\beta_{i'}}\inf_{s_i\in\mathcal{S}\backslash s_{i'}}\{\beta_i\},
	\end{align}
	which basically tells us that when $s_{i'}$ undergoes a path-loss at least 3 dB greater than the experienced by the remaining EH nodes, the optimum energy beamformer is at least half of the time given by \eqref{wj} since $\Omega>1$. 
\end{remark}
	\subsection{CSI-free WET}\label{free}
	Several CSI-free powering schemes have been recently proposed and analyzed in \cite{Lopez.2019_CSI,Lopez.2020}, e.g., $\mathrm{OA}$, $\mathrm{AA-SS}$, $\mathrm{AA-IS}$ and $\mathrm{SA}$. In Section~\ref{introduction}, we highlighted their main characteristics and argued why we adopt the $\mathrm{SA}$ scheme as our CSI-free scheme in this paper. Summarizing, the reasons are three-fold: i) among the schemes taking advantages of the spatial resources, $\mathrm{SA}$ exhibits a homogeneous performance over the space, which is not sensitive to the different mean phases of the LOS channel component; ii) it is more suitable than the $\mathrm{AA-IS}$ scheme for powering devices far from the HAP; and iii) it allows a better coupling to the co-located information transmission processes since only one antenna may be transferring energy in the downlink while the remaining may be receiving uplink information.
	
	The following results characterize the statistics of the incident RF power under the $\mathrm{SA}$ scheme and provide energy outage expressions.
	\begin{proposition}
		The distribution of the incident RF power at $s_{i'}$ under $\mathrm{SA}$ is given by
		\begin{align}
		E_{i'}^\mathrm{rf}&=\frac{P\beta_{i'}}{M}||\mathbf{\tilde{h}}_{i'}^{(d)}||^2
		\sim \frac{P\beta_{i'}}{2M(1+\kappa)}\chi^2(2M,2M\kappa).\label{Eirf}
		\end{align}
	\end{proposition}
\begin{proof}
	In our setup, the adoption of the SA scheme implies that each transmit antenna is active during $T_c/M$ seconds, while the remaining $M\!-\!1$ antennas function as receive antennas, i.e., $M_t\!=\!1,\ M_r\!=\!M\!-\!1$. Then, we can directly state \eqref{Eirf} by exploring the connection to \eqref{Ei'} and defining $\mathbf{\tilde{h}}_{i'}^{(d)}\in\mathbb{C}^{M\times 1}$ since all antennas  transmit during a coherence block but not simultaneously.
\end{proof}
	\begin{corollary}\label{cor1}
		Consequently, by taking $\xi_\mathrm{csi}^{(d)}\leftarrow 0$, $P\leftarrow P/M$ and $M_t\leftarrow M$, we conclude that \eqref{per} and \eqref{cfO} hold, but herein as exact and approximate energy outage expressions under periodic and Poisson traffic,  respectively, instead of lower bounds.
	\end{corollary}
	\begin{remark}\label{re4}
		Under the CSI-based scheme, the average harvested energy can be up to $M_t$ times greater than under $\mathrm{SA}$ scheme, for which $\mathbb{E}[E_{i'}^\mathrm{rf}]=P\beta_{i'}$; however, the diversity gain of $\mathrm{SA}$ is $M/M_t^{\mathrm{mrt}}>1$ greater since all antennas contribute. Additionally, note that differently from $\mathrm{SA}$,  the WET performance under the CSI-based scheme is affected in practice by the CSI inaccuracy.
	\end{remark}
\begin{remark}\label{re4p5}
	Assuming an homogeneous deployment around the HAP, we can assure that as the number of EH devices increases, the average harvested energy under the CSI-based EB decreases approaching that attained under  $\mathrm{SA}$. This is because the maximum gap between these schemes is $M_t$ (see Remark~\ref{re4}) and happens when the CSI-based scheme matches the MRT beamforming. However, as the number of devices increases, the greater the chances of $\Omega$ in \eqref{Ec} be smaller since $\inf_{s_i\in\mathcal{S}\backslash s_{i'}}\{\beta_i\}$ is expected to decrease,  thus MRT is less often the optimum CSI-based EB.
\end{remark}
	\begin{table*}[!t]
		\centering
		\caption{Main system performance characteristics under the considered CSI-based and CSI-free WET schemes}
		\begin{tabular}{lccccc}
			\toprule
			\textbf{Schemes} & \textbf{No. Tx. Antennas} &  \textbf{Average EH Gain} & \textbf{EH Diversity} & \textbf{No. Rx. Antennas} & \textbf{Energy Requirements} \\
			\midrule
			CSI-based  & $M_t$ & $\le M_t$ & $M_t$ & $M-M_t$ & Moderate$-$High\\ 
			CSI-free (SA)  & $1$  & $1$ &$M$ & $M-1$ & Low$-$Moderate\\		
			\bottomrule
		\end{tabular}\label{table1}
	\end{table*}
	
	A summary on the system performance characteristics under the CSI-based and CSI-free WET schemes is presented in Table~\ref{table1}. Notice that the average EH gain is counted as $\mathbb{E}[E^\mathrm{rf}_{i'}]/(P\beta_{i'})$, while the energy requirement field accounts for all energy consumption sources including the uplink CSI-acquisition procedure which is required for both analyzed schemes.
	\section{Wireless Information Transmission}\label{WIT}
	As commented in Section~\ref{system}, at some points, the EH devices require sending short data messages of $k$ bits/Hz over blocks of $t$ seconds to the HAP. The HAP uses $M_r$ antennas to decode the arriving messages and resolve possible simultaneous transmissions.   We assume the uplink CSI, which is needed to implement ZF or MMSE linear decoding schemes adopted here, is perfectly acquired  at the HAP (jointly with devices' detection in case of Poisson traffic) when receiving the pilot signals sent by the active EH devices, nonetheless, some illustrative results on the degenerative effect of imperfect uplink CSI acquisition are discussed in Section~\ref{results}.	
	Additionally, we consider orthogonal pilot signals.
	Note that non-orthogonal allocations, e.g., as in \cite{Shao.2020}, would require very sporadic activation profiles and/or a very large antenna array at the HAP for efficient detection and CSI acquisition, and it is left for future work together with the analysis under non-coherent decoding schemes and a rigorous analytical treatment of imperfect CSI acquisition.
	
	We consider an uninterrupted downlink WET, while now and then a subset of the devices interrupt their harvesting process  to send uplink data. The self-interfering powering signals, traveling through the channels between the $M_t$ transmit antennas and $M_r$ receive antennas when using either the CSI-based or CSI-free WET scheme, are assumed to be perfectly canceled  via Successive Interference Cancellation (SIC) techniques; while we discuss the impact of imperfect SIC through some numerical results in Section~\ref{generalP}. Note that SIC techniques may include analog and digital processes, and can even benefit from the fact that the powering signals may be chosen deterministically under the CSI-free WET scheme.
	
	Finally, under SA we assume that the transmit slots are scheduled such that no antenna switching occurs during an actual uplink transmission, which would complicate the information decoding procedures. Next, we analyze the WIT performance under the considered traffic profiles.
	\subsection{WIT under periodic traffic}\label{periodic}
	As commented in Subsection~\ref{txmodel}, the maximum number of concurrent transmissions is deterministically $\lceil\mathrm{S}/\lfloor t_s/t\rfloor\rceil$, thus, the same number of orthogonal pilot signals (and pilot symbols per signal) for uplink CSI estimation is required. Then, under periodic traffic, $\xi_\mathrm{csi}^{(u)}$ can be broken approximately into
	\begin{align}
	\xi_\mathrm{csi}^{(u)}\approx\lceil\mathrm{S}/\lfloor t_s/t\rfloor\rceil \tilde{\xi}_0,
	\end{align}
	where $\tilde{\xi}_0$ is the per-symbol pilot energy. Next, we investigate the outage performance of the data transmission phase.
	\subsubsection{Signal model}
	At the HAP, the data signal received after each transmission is given by
	\begin{align}
	\mathbf{y}=\mathbf{H}^{(u)}\mathbf{P}_{\!\beta}^{1/2}\mathbf{x}+\mathbf{w},\label{y}
	\end{align}
	where the $j-$th column of $\mathbf{H}^{(u)}$ is $\mathbf{h}_j^{(u)}$ and consequently such matrix has dimension $M_r\times \mathrm{S}$, $\mathbf{P}_{\!\beta}=\mathrm{diag}\big(\{p_i\beta_i\}\big)$, $\mathbf{x}\in\mathbb{C}^{\mathrm{S\times 1}}$ is the normalized vector of the normal signals transmitted by the $\mathrm{S}$ devices, and $\mathbf{w}\sim\mathcal{CN}(\mathbf{0},\sigma^2\mathbf{I}_{M_r\times M_r})$ is the Additive White Gaussian Noise (AWGN) vector at the $M_r$ antennas. 
	If $s_i$ is not active in a given transmission slot of duration $t$, we consider that the respective entries in $\mathbf{H}^{(u)}$, $\mathbf{P}_{\!\beta}$ and $\mathbf{x}$ are zero. Consequently the number of non-zero columns of $\mathbf{H}^{(u)}$ is at most $\lceil\mathrm{S}/\lfloor t_s/t\rfloor\rceil$, which matches also the maximum number of non-zero rows and columns of $\mathbf{P}_{\!\beta}$, and the number of non-zero elements of $\mathbf{x}$. 
	Finally, the equalizer $\mathbf{Q}\in\mathbb{C}^{\mathrm{S}\times M_r}$ at the receiver decouples the transmitted data streams such that its output is given by
	\begin{align}
	\mathbf{y}_{\mathrm{out}}=\mathbf{Q}\mathbf{y}=\mathbf{Q}\mathbf{H}^{(u)}\mathbf{P}_{\!\beta}^{1/2}\mathbf{x}+\mathbf{Q}\mathbf{w}.\label{yout}
	\end{align}
	\subsubsection{ZF}
	The ZF equalizer is 
	\begin{align}
	\mathbf{Q}^\mathrm{zf}&=\Big(\big(\mathbf{H}^{(u)}\mathbf{P}_{\!\beta}^{1/2}\big)^H\mathbf{H}^{(u)}\mathbf{P}_{\!\beta}^{1/2}\Big)^{-1}\big(\mathbf{H}^{(u)}\mathbf{P}_{\!\beta}^{1/2}\big)^H\nonumber\\
	&=\mathbf{P}_{\!\beta}^{-1/2}\Big(\mathbf{H}^{(u)H}\mathbf{H}^{(u)}\Big)^{-1}\mathbf{H}^{(u)H},
	\end{align}
	and by substituting it into \eqref{yout} yields
	\begin{align}
	\mathbf{y}_{\mathrm{out}}^\mathrm{zf}=\mathbf{x}+\mathbf{P}_{\!\beta}^{-1/2}\Big(\mathbf{H}^{(u)H}\mathbf{H}^{(u)}\Big)^{-1}\mathbf{H}^{(u)H}\mathbf{w}.
	\end{align}
	Then, the instantaneous Signal to Interference-plus-Noise Ratio (SINR) of the output stream corresponding to the one transmitted by $s_i$ is given by 
	\begin{align}
	\gamma_i^\mathrm{zf}&=\frac{1}{\Big[\Big(\big(\mathbf{H}^{(u)}\mathbf{P}_{\!\beta}^{1/2}\big)^H\mathbf{H}^{(u)}\mathbf{P}_{\!\beta}^{1/2}\Big)^{-1}\Big]_{i,i}\!\sigma^2}\nonumber\\
	&=\!\frac{1}{\Big[\mathbf{P}_{\!\beta}^{-1/2}\mathbf{Z}\mathbf{P}_{\!\beta}^{-1/2}\Big]_{i,i}\!\sigma^2}\!=\!\frac{p_i\beta_i}{\sigma^2}Z^{\mathrm{zf}},\label{g1}
	\end{align}
	where $Z^{\mathrm{zf}}=1/Z_{i,i}$ with $\mathbf{Z}=\big(\mathbf{H}^{(u)H}\mathbf{H}^{(u)}\big)^{-1}$. 
	\begin{remark}\label{re5}
	Notice that for Rayleigh fading, i.e., $\kappa=0$, $\mathbf{Z}$ has the central inverse Wishart distribution, which for the case of $M_r$ greater than the number of data streams $N$, yields to $Z^\mathrm{zf}\sim 2\chi^2(2(M_r-N+1))$. Meanwhile, the analysis under Rician fading is encumbered by the noncentrality of the Wishart distribution of $\mathbf{Z}^{-1}$. The usual approach in such case lies in approximating the noncentral Wishart distribution by the virtual central Wishart distribution as summarized in \cite{Siriteanu.2012}. In any case, the analysis is cumbersome, specially for the general scenario where $M_r\ge N$ does not need to necessarily hold, thus, we take no further steps to characterize the distribution of $Z^{\mathrm{zf}}$. 
	\end{remark}
	\subsubsection{MMSE} 
	The MMSE equalizer is 
	\begin{align}
	\mathbf{Q}^\mathrm{mmse}&\!\!=\!\Big(\!\mathbf{P}_{\!\beta}^{1/2}\mathbf{H}^{(u)H}\mathbf{H}^{(u)}\mathbf{P}_{\!\beta}^{1/2}\!+\!\sigma^2\mathbf{I}\Big)^{\!-1}\!\!\big(\mathbf{H}^{(u)}\mathbf{P}_{\!\beta}^{1/2}\big)^{\!H}\!\!,
	\end{align}
	while the corresponding component for decoding the $i-$th data stream is given by
	\begin{align}
	\mathbf{q}^\mathrm{mmse}_i=\Big(\sigma^2\mathbf{I}+\sum_{j\ne i}^{\mathrm{S} }\beta_j p_j\mathbf{h}_j^{\!(u)}\mathbf{h}_j^{\!(u)H}\Big)^{-1}\!\mathbf{h}_i^{\!(u)}.
	\end{align}
	Then, the corresponding instantaneous SINR is given by
	\begin{align}
	\gamma_i^\mathrm{mmse}&=\frac{\beta_ip_i}{\sigma^2}Z^\mathrm{mmse},\label{g2}
	\end{align}
	where $Z^\mathrm{mmse}=\mathbf{h}_i^{\!(u)H}\Big(\mathbf{I}+\sum_{j\ne i}^{\mathrm{S} }\frac{\beta_j p_j}{\sigma^2}\mathbf{h}_j^{\!(u)}\mathbf{h}_j^{\!(u)H}\Big)^{-1}\!\mathbf{h}_i^{\!(u)}$. 
	\begin{remark}\label{re6}
		Even in the simplest scenario with Rayleigh fading, equal per-user SNR, and $M_r\!\ge\! N$, the distribution of $Z^\mathrm{mmse}$ is cumbersome as corroborated in \cite{Gao.1998,Lim.2019}. This, and the fact that for a more general scenario there is no closed-form expression for the PDF and CDF of $Z^\mathrm{mmse}$.
	\end{remark}
	\subsubsection{Information outage performance}
	For the sake of fairness, we assume that those devices with the most similar path losses are scheduled for concurrent transmission. This is possible under periodic traffic, which is deterministic by nature. Let us sort the devices according to their path loss such that $s_1$ is the device with the smallest attenuation, while $s_\mathrm{S}=s_{i'}$ is the device under the greatest path loss. Now, we evaluate the information outage performance at $s_{i'}$ in order to get a bound on the performance of any node in the network\footnote{Such bound is expected to hold under the assumption of equal devices' transmit power, or a power allocation such that a greater attenuation implies a smaller transmit power. While the latter seems odd at first sight since the farthest node is usually allowed to transmit with greater power in traditional cellular networks, it is not the case in WPCNs where the farthest device harvests also less energy.}. Thus, we have that $N=\lceil\mathrm{S}/\lfloor t_s/t\rfloor\rceil$, while \eqref{y} and subsequent derivations can be compacted by eliminating the zero-rows/columns of $\mathbf{H}^{(u)},\ \mathbf{P}_\beta$, e.g. $\mathbf{H}^{(u)}\in \mathbb{C}^{M_r\times N}$, $\mathbf{P}_\beta=\mathrm{diag}\big(\{p_i\beta_i\}\big),\ i\!=\!\mathrm{S}\!-\!N\!+\!1,\cdots,\mathrm{S}$, and reducing the dimension of vector $\mathbf{x}$, i.e., $\mathbf{x}\in \mathbb{C}^{N\times 1}$.  Then,
	\begin{align}
	\sup_i\{\mathcal{O}_i^{(u)}\}&= \mathcal{O}_{i'}^{(u)}=\mathbb{P}\big[\log_2(1+\gamma_{i'})<k/t\big]\nonumber\\
	&=\mathbb{P}\big[\gamma_{i'}<2^{k/t}-1\big]\nonumber\\
	&=F_Z\bigg(\frac{(2^{k/t}-1)\sigma^2}{\beta_{i'}p_{i'}}\bigg),\label{gam}
	\end{align}
	which comes from using \eqref{g1} and \eqref{g2} such that $Z\in\{Z^\mathrm{zf},Z^\mathrm{mmse}\}$. Since the distribution of $Z$ is intractable (see Remarks~\ref{re5} and \ref{re6}), we evaluate \eqref{gam} by drawing samples of $Z$ via Monte Carlo. Notice that this constitutes a semi-analytical computation of the information outage since the traffic is not simulated, thus, it is much easier and more efficient to compute than relying on a pure Monte Carlo approach \cite{Rubinstein.2016}.
	\subsection{WIT under Poisson traffic}\label{Poisson}
	Concurrent transmissions happen randomly under Poisson traffic. Therefore, there is no way of completely avoiding  the pilot collisions unless all devices are allocated orthogonal pilot sequence. However, this can be extremely energy-costly for large $\mathrm{S}$ since $\xi_\mathrm{csi}^{(u)}=\mathrm{S}\tilde{\xi}_0$, where $\tilde{\xi}_0$ was defined in the previous subsection. To overcome this, we herein allow collisions to occur with a probability not greater than $\varepsilon$, which is a system parameter to be efficiently designed.
	\subsubsection{Collision probability}
	The probability that a given device $s_i$ is active at a certain time slot is given by
	\begin{align}
	\mathbb{P}\big[s_i\in \tilde{\mathcal{S}}\big]&=\frac{1}{\mathbb{E}_V[vT_c/t]}=\frac{t/T_c}{\mathbb{E}_V[v]}
	=\frac{t}{T_c}\big(1-e^{-\lambda}\big), \label{ps}
	\end{align}
	where $\tilde{\mathcal{S}}\subseteq\mathcal{S}$ denotes the set of active devices in such a time slot, and the last equality comes from computing $\mathbb{E}_V[v]$, which is given in \eqref{Ev} in Appendix~\ref{App_A}.
	\begin{remark}\label{re7}
		Notice that the subset $\tilde{\mathcal{S}}\subseteq\mathcal{S}$ of active devices is random under Poisson traffic, and also its cardinality $N=|\tilde{\mathcal{S}}|$, which is a Binomial RV with parameters $\mathrm{S}$ and $\frac{t}{T_c}\big(1-e^{-\lambda}\big)$.
	\end{remark}
	
    The collision probability is characterized in the following result.
	\begin{theorem}\label{the4}
		Assuming certain device is active, its associated collision probability is given by
		\begin{align}
		\mathcal{O}_\mathrm{col}&=1-\frac{\frac{\mathrm{S}t}{LT_c}(1-e^{-\lambda})\Big(1-\frac{t}{LT_c}(1-e^{-\lambda})\Big)^{\mathrm{S}-1}}{1-\Big(1-\frac{t}{LT_c}(1-e^{-\lambda})\Big)^{\mathrm{S}}},\label{ocol}
		\end{align}
		where $L$ is the number of orthogonal pilot sequences/symbols and $L/\mathrm{S}$ is the pilot reuse factor.
	\end{theorem}
\begin{proof}
	See Appendix~\ref{App_C}. \phantom\qedhere
\end{proof}
	Then, we must choose $L$ such that $\mathcal{O}_\mathrm{col}\le\varepsilon$. However, notice that if such $L$ is greater than $\mathrm{S}$, it is preferable to deterministically assign  one unique pilot sequence to each user, thus,  avoiding  the collisions completely. Therefore, the optimum $L$ given $\varepsilon$ is given by
	\begin{align}
	L^*=\left\{\begin{array}{lll}
	L_0, &\text{if}\ L_0<\mathrm{S}&\rightarrow \mathcal{O}_\mathrm{col}\ \text{given in \eqref{ocol}} \\
	\mathrm{S}, &\text{otherwise}&\rightarrow \mathcal{O}_\mathrm{col}=0
	\end{array}\right.,\label{Leq}
	\end{align}
	where 
	\begin{align}
	L_0 = \inf_{L\in\mathbb{Z}^+,\mathcal{O}_\mathrm{col}\le \varepsilon} L.\label{problem}
	\end{align}
	In Appendix~\ref{App_D}, we illustrate a simple procedure for solving \eqref{problem}. Finally,
	\begin{align}
	\xi_\mathrm{csi}^{(u)}=L^*\tilde{\xi}_0.
	\end{align}
	Next, we investigate the outage performance of the data transmission phase.
	\subsubsection{Information outage performance}
	Herein, we need to consider the pilot collision events and the outages due to decoding errors.
	Since $\mathcal{O}_\mathrm{col}$ takes into account the events related to the collided $s_{i'}$'s transmissions, we are now interested on the event where $s_{i'}$ operates without collision while the remaining IoT sensors in $\tilde{S}$ may or may not be colliding. 
	Consequently, we now have that
	\begin{align}
	\sup_i&\{\mathcal{O}_i^{(u)}\}\nonumber\\
	&\stackrel{(a)}{=} \mathcal{O}_\mathrm{col}\!+\!\big(1\!-\!\mathcal{O}_\mathrm{col}\big)\mathbb{E}_{\tilde{\mathcal{S}}|N\ge 1}\bigg[F_{Z|\tilde{\mathcal{S}},N\ge 1}\Big(\frac{(2^{k/t}\!-\!1)\sigma^2}{\beta_{i'}p_{i'}}\Big)\bigg]\nonumber\\
	&\stackrel{(b)}{\approx} \varepsilon\!+\!(1-\varepsilon)\mathbb{E}_{\tilde{\mathcal{S}}|N\ge 1}\bigg[F_{Z|\tilde{\mathcal{S}},N\ge 1}\Big(\frac{(2^{k/t}\!-1)\sigma^2}{\beta_{i'}p_{i'}}\Big)\bigg],\label{gam2}
	\end{align}
	where the last line comes from using $\mathcal{O}_{\mathrm{col}}\approx\varepsilon$, which holds as long as the system is properly designed. The latter term in both $(a)$ and $(b)$ can be easily  evaluated by 
	\begin{enumerate}
		\item generating a sample $N$ conditioned on $N\ge 1$ (see Remark~\ref{re7});
		\item drawing $N-1$ elements from $\mathcal{S}\backslash s_{i'}$ to conform the set of interfering devices;
		\item evaluating \eqref{gam} for such configuration;
		\item averaging \eqref{gam} over many possible realizations of $N$.
	\end{enumerate}
	Note that setting the target collision  probability $\varepsilon$ for optimum system performance is a challenging task since the last term of \eqref{gam2} intricately depends  on $\varepsilon$, hence, numerical analysis seems unavoidable and is carried out in the next section.
	\begin{remark}\label{re8}
	Summarizing, the key differences between the WIT analysis under periodic and Poisson traffic are: i) the number of concurrent transmissions is either deterministic (under periodic traffic) or random (under Poisson traffic); ii) the required number of pilot sequences under periodic traffic is fixed and matches the number of concurrent transmissions, while under Poisson traffic more pilot sequences are required to mitigate random pilot collisions; iii) previous issue imposes an important design challenge in terms of choosing the appropriate number of pilot sequences in Poisson traffic scenarios since the more pilot sequences are used, the smaller the collision probability but the greater the training power consumption and the chances of energy outage.
	\end{remark}	
	\section{Numerical Results}\label{results}
	In this section, numerical examples are provided to corroborate our study and evaluate the suitability of the CSI-free or CSI-based WET schemes. We assume $\lambda=T_c/t_s$ for a fair comparison between the periodic and Poisson traffic profiles. Also, the HAP has a maximum transmit power $P=10$ W, and its associated devices are distributed around in a $12$ m-radius circular area as shown in Fig.~\ref{Fig2}. Specifically, we consider the EH devices are at $2,4,6,8,10$ and $12$ m from the HAP, while the number of devices in each sub-circumference is proportional to its length, thus, devices are approximately uniformly distributed in the coverage area. 
	\begin{remark}\label{re9}
		According to the adopted deployment scenario and the results related to Theorem~\ref{the3} and Remark~\ref{re3}, around half of the EH devices (those at 10m and 12m from the HAP) are expected to fully determine the optimum EB most of the time. Consequently the HAP requires coordinating approximately  $\mathrm{S}/2$ uplink transmissions for acquiring an effective downlink CSI of the network.
	\end{remark}

    Based on Remark~\ref{re9} and considering a certain EH circuitry being adopted, $\eta$ must be chosen such that it matches the conversion efficiency at which devices at $10-12$ m from the HAP operate on average. This can be computed based on the average RF receive power, and consequently may differ under the CSI-based and CSI-free schemes when powering a relatively small set of devices. However, herein we adopt the same EH conversion efficiency of $\eta=0.25$ in both cases for simplicity.
   	\begin{figure}[t!]
    	\centering  
    	\includegraphics[width=0.35\columnwidth,trim= 5 0 20 0]{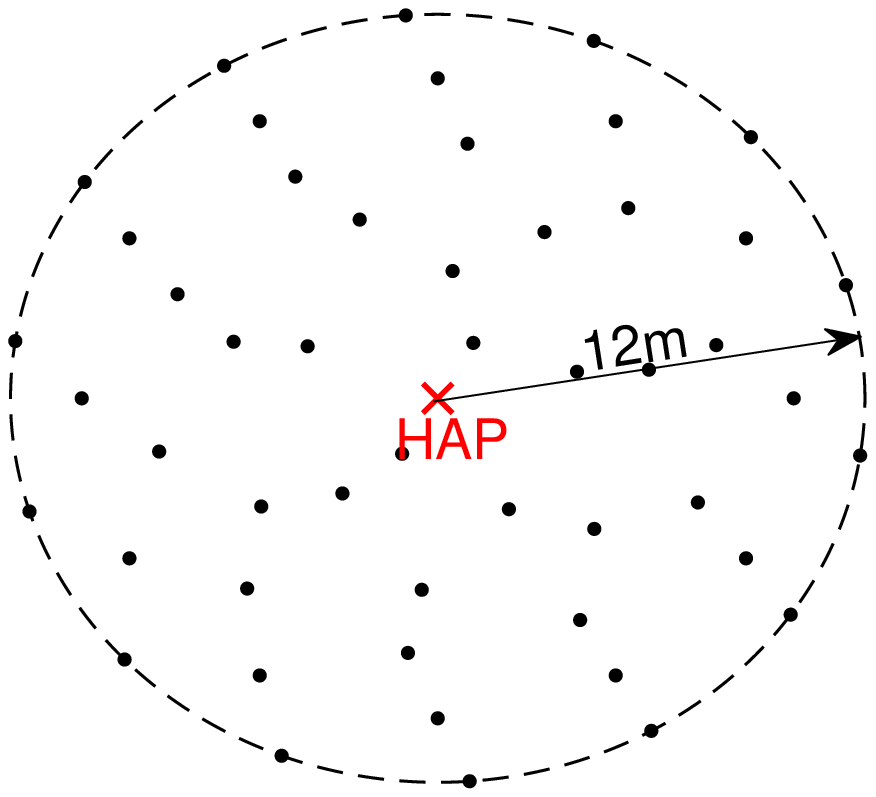}\qquad\qquad\ \  \includegraphics[width=0.35\columnwidth,trim= 5 0 20 0]{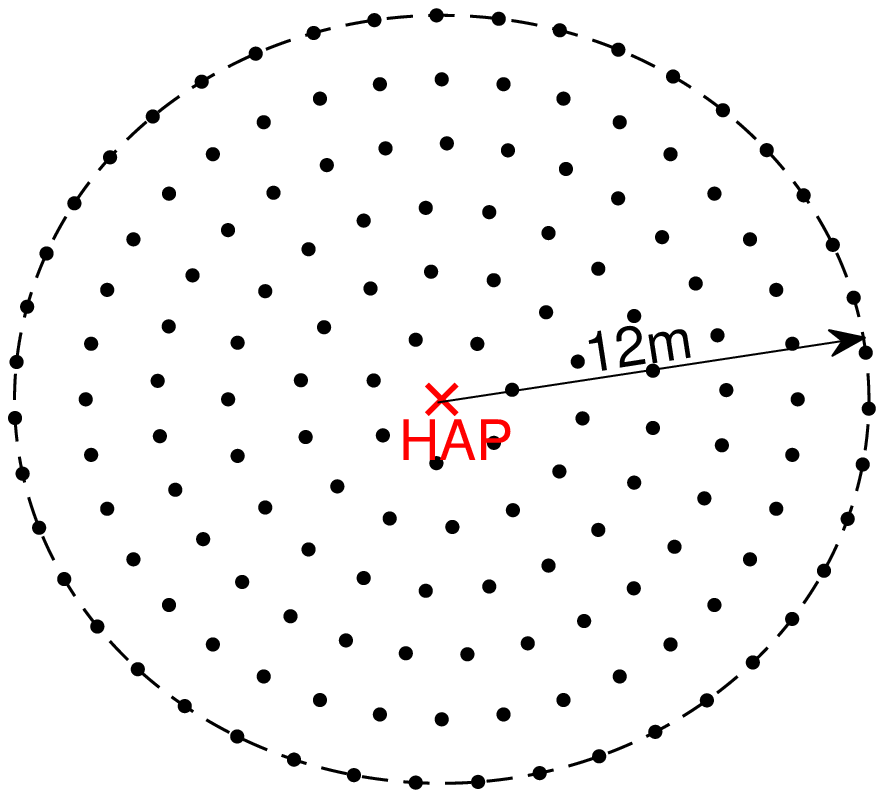}		
    	\caption{Example deployments: $(a)$ $\mathrm{S}=50$ (left) and $(b)$ $\mathrm{S}=150$ (right).}		
    	\label{Fig2}
    \end{figure}

	Channels remain static for $T_c=400$ ms, and are subject to a log-distance path-loss model with exponent $2.7$ with non-distance dependent losses of $16$ dB. Thus, the average channel gain corresponding to the $s_i$'s uplink/downlink channels is given by $\beta_i=10^{-1.6}\times d_i^{-2.7}$, where $d_i$ is the distance between $s_i$ and the HAP. The noise power at the HAP's receive antennas is assumed to be $\sigma^2=-94$ dBm. The circuit power consumption and transmit power of the EH devices is set to $20\ \mu$W and $200\ \mu$W, respectively. Devices are required to transmit each message in the uplink within $t=20$ ms time window ($M_0=20$ and we limit our analysis to $M\le 20$). 
	Additionally, unless stated otherwise, we set $S=100$ to account for a massive deployment ($\sim 0.22$ devices/$\mathrm{m}^2$)\footnote{The projections towards 6G point to challenging scenarios with up to $10$  devices/$\mathrm{m}^2$ \cite{Mahmood.2019,Mahmood.2020}, thus, the considered deployment is not as massive as it can be 10 years from now. Obviously, the larger the number of EH devices is, the more beneficial the analyzed CSI-free schemes are when compared to the traditional CSI-based schemes \cite{Lopez.2019_CSI}.}, $\kappa=5$ to account for some LoS, and $k= 10^{-3}$ bits/Hz to account for low-rate transmissions as typical in MTC. Finally, $M=6$, $t_s=1.6$ s, $\xi_0=-20$ dBm and $\tilde{\xi}_0=-30$ dBm.
	\subsection{On the WET performance}
 \begin{figure}[t!]
	\centering
	\includegraphics[width=0.95\columnwidth]{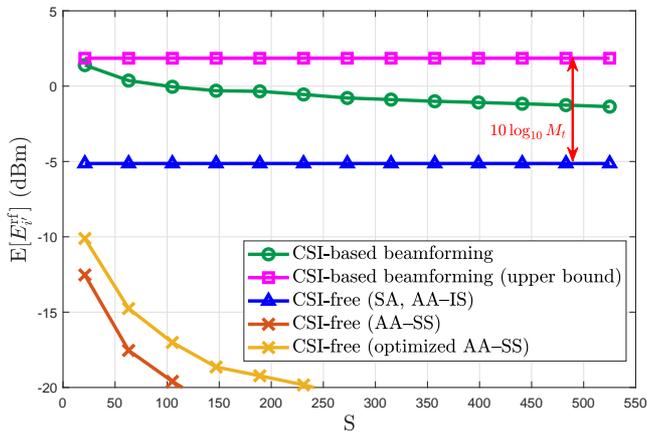}
	\caption{Worst-case average RF energy availability as a function of $\mathrm{S}$. We set $M_t=5$ for the CSI-based schemes and CSI-free ($\mathrm{AA-SS}$). The results corresponding to $\mathrm{AA-SS}$ are drawn assuming the PB is equipped with a uniform linear array with half-wavelength spaced antenna elements.}		
	\label{Fig3}
\end{figure}
	Herein, we investigate the EH performance of the farthest node $s_{i'}$ under CSI-free and CSI-based WET schemes. Such node performs the worst in the network, thus, we can guarantee a minimum level performance for the entire set of devices. Fig.~\ref{Fig3} shows the average RF energy availability as a function of the number of devices. Notice that the performance of the CSI-free SA and $\mathrm{AA-IS}$ schemes is independent of the number of devices since they provide a uniform performance along the area. This is different from the CSI-based approach for which the beamforming gains decrease as $\mathrm{S}$ increases, and the CSI-free $\mathrm{AA-SS}$ scheme, which favors certain spatial directions only and its use is recommended in clustered setups where positioning information is available \cite{Lopez.2020,LopezMahmood.2020}.
	Although both SA and $\mathrm{AA-IS}$ perform similarly in the WET phase,  SA outperforms $\mathrm{AA-IS}$ in terms of overall performance by allowing the use of more receive antennas in the WIT phase. Please refer to \cite{Lopez.2020,LopezMahmood.2020} for an extensive comparison between $\mathrm{SA}$ and other state-of-the-art CSI-free schemes in terms of harvested energy and/or RF energy availability at the EH devices, while hereinafter we just discuss the performance of SA as the CSI-free WET scheme. 
	Obviously, only when $\mathrm{S}\rightarrow\infty$, both SA and the CSI-based scheme converge to the same performance in terms of average RF energy  supply (see Remark~\ref{re4p5}).
	The best possible performance is when the farthest node is powered via a  CSI-based MRT scheme, and still such node is sufficiently far such that it keeps performing the worst in the network. Therefore, a MRT under such circumstance provides an upper bound performance, which is always $10\log_{10}M_t$ dB greater than SA's (see Remark~\ref{re4}), as also illustrated in Fig.~\ref{Fig3}. From now on, we only show the results corresponding to the MRT bound since for the exact CSI-based performance, $\mathbf{P}$ in \eqref{P} requires to be repeatedly solved, which is extremely costly as illustrated in Remark~\ref{re2}.
	While  the CSI-based scheme always outperforms the SA CSI-free scheme in terms of average RF energy  supply, that is not longer the case when analyzing the EH performance in terms of energy outage probability. On one hand, the EH diversity is smaller in case of the CSI-based scheme as shown in Table~\ref{table1}. On the other hand, the CSI-based scheme introduces additional sources of energy consumption, which is accounted in the term $\xi_\mathrm{csi}^{(d)}$ and depends specifically on $M_t$ and $\xi_0$ as stated in \eqref{csid}. 
	\begin{figure}[t!]
	\includegraphics[width=0.95\columnwidth]{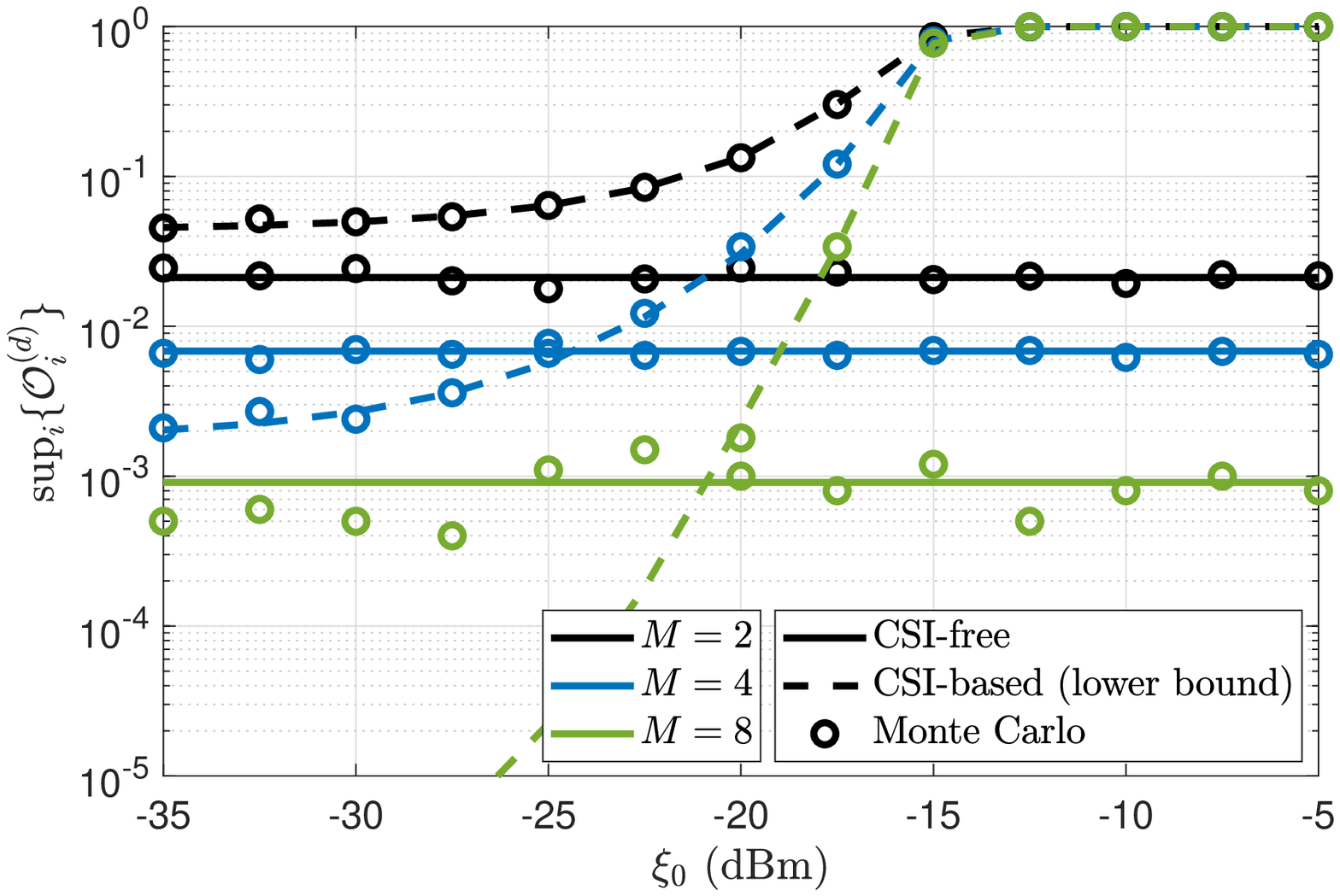}\vspace{3mm}\\
	\includegraphics[width=0.95\columnwidth]{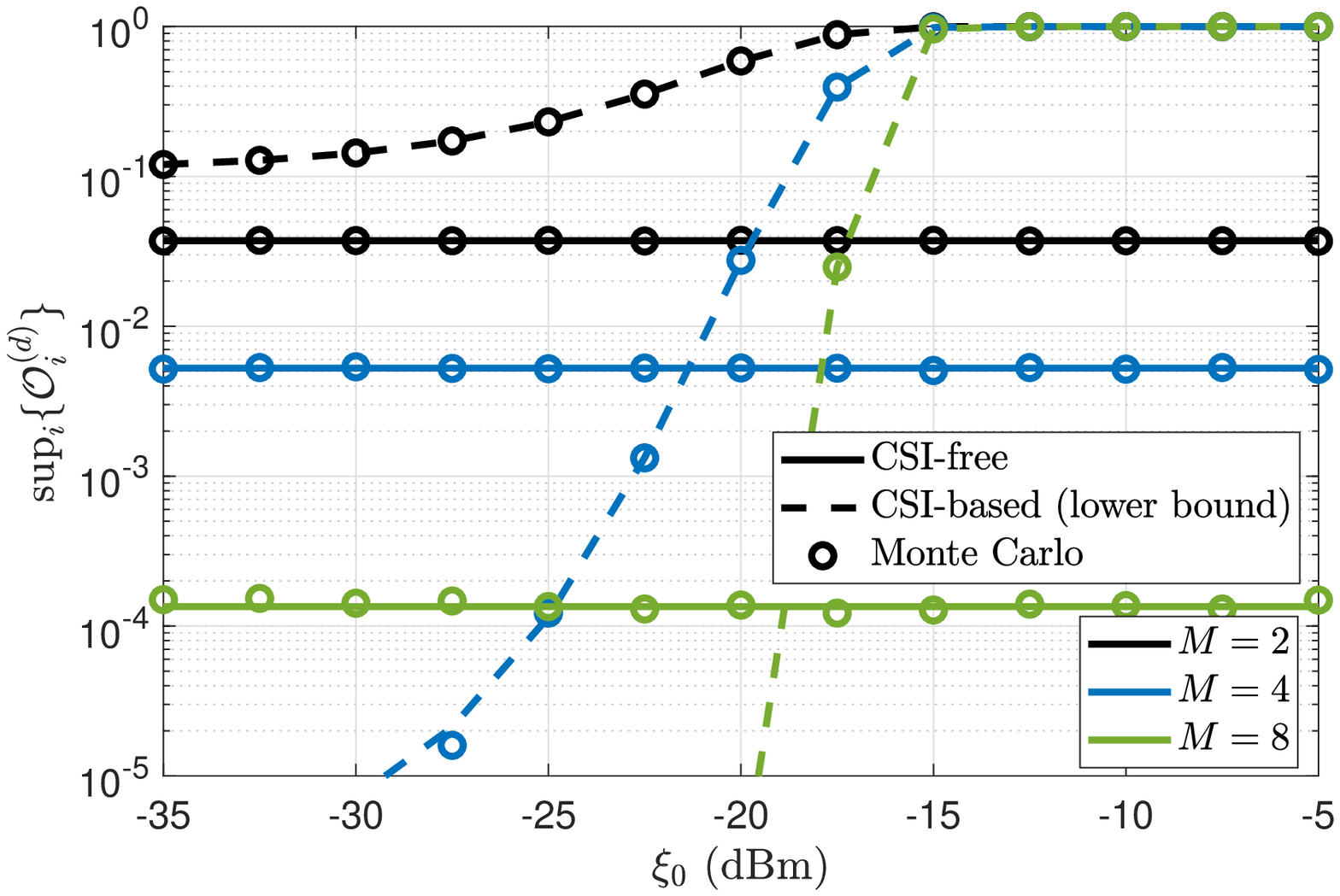}
	\caption{Worst-case energy outage probability as a function of $\xi_0$. $a)$ Poisson traffic, $\varepsilon=10^{-1}$, $p_c=20\mu$W (top); and $b)$ Periodic traffic, $p_c=50\mu$W. For the CSI-based schemes we set $M_t=M_r=M/2$ (bottom).}		
	\label{Fig4}
\end{figure}
	
	Fig.~\ref{Fig4} shows the worst-case energy outage probability as a function of $\xi_0$ for both Poisson and periodic\footnote{We used a much more restrictive circuit power consumption level in case of periodic traffic to better visualize the outage performance in a range of values that can be corroborated via Monte Carlo simulation. This was done only for Fig.~\ref{Fig4}.} traffics when half of the antennas are used in downlink/uplink under the CSI-based scheme. The performance of the CSI-free SA scheme remains obviously constant, while  the MRT gains from CSI disappear quickly as $\xi_0$ increases. However, notice that even when a greater $M_t$ increases the CSI acquisition costs, it is still more advantageous than costly for the CSI-based scheme, since a greater $M_t$ enlarges the $\xi_0$ region for which the CSI-based scheme is preferable. Finally, notice that under periodic traffic the performance is much better than when devices activate randomly according to a Poisson process.
	\subsection{On the WIT performance}\label{WITresults}
	Herein, we investigate the communication performance of the farthest node $s_{i'}$ when powered via either CSI-free or CSI-based WET schemes. We evaluate the worst-case information outage probability given a communication attempt. Specifically, Fig.~\ref{Fig5}a shows the performance degradation as the number of EH devices increases when the information decoding is done via MMSE and ZF techniques under Poisson traffic. In general, ZF is known to approach the MMSE performance at high SINR, but notice that here the MMSE scheme outperforms significantly the ZF scheme since the operation is at relatively small SINRs, because of the low-power low-rate transmissions. Therefore, operating under the MMSE decoding scheme is highly recommended, and hereinafter we only present the results related to MMSE. Note that the performance improves as $\varepsilon$ decreases, since the collision probability decreases. However, as $\varepsilon$ decreases, the changes of energy outage increase, which is not considered here, and may degrade the overall system performance as discussed in Remark~\ref{re8} and illustrated later in the next subsection. Although not shown in Fig.~\ref{Fig5}a, it is worth highlighting that the performance improves as the number of receive antennas increases. All previous trends hold  also under periodic traffic although with better relative performance as observed in Fig.~\ref{Fig5}b. This figure shows the information outage as a function of the average number of coherence time intervals between the transmission of consecutive messages from each device (the periodicity in case of periodic traffic and the inverse of traffic rate in case of Poisson traffic). As such time increases, the performance improves since the number of concurrent transmissions decreases. In case of deterministic traffic, the maximum performance is attained when $N=1$, which occurs when $t_s\ge 2$ s, as corroborated in the figure. In case of Poisson traffic, the chances of no concurrent transmissions vanish only when $L^*=\mathrm{S}$, which tends to happen as $\varepsilon$ decreases and/or $T_c/\lambda$ increases as shown in the figure. 
		\begin{figure}[t!]
		\includegraphics[width=0.95\columnwidth]{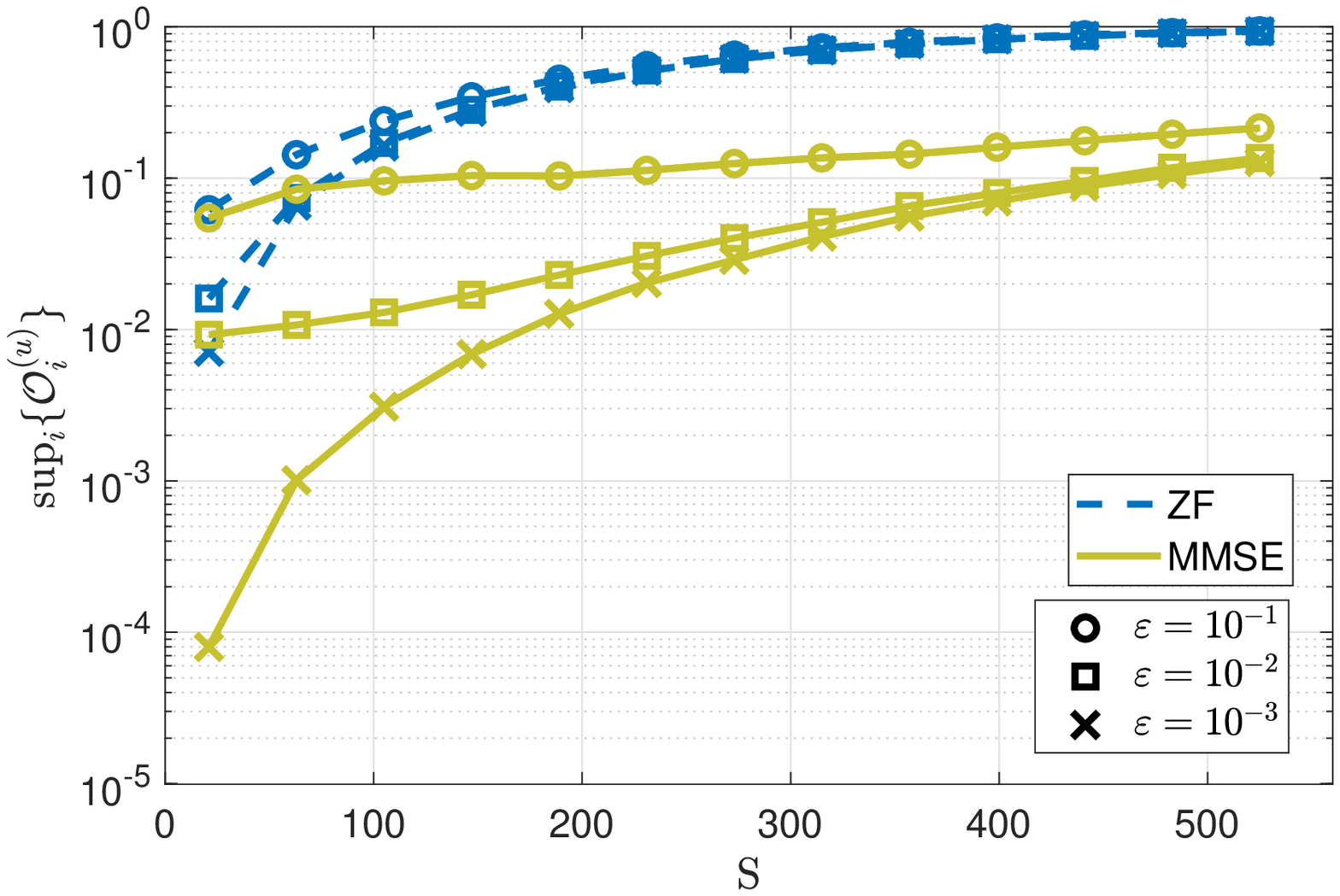}\vspace{3mm}\\
		\includegraphics[width=0.965\columnwidth]{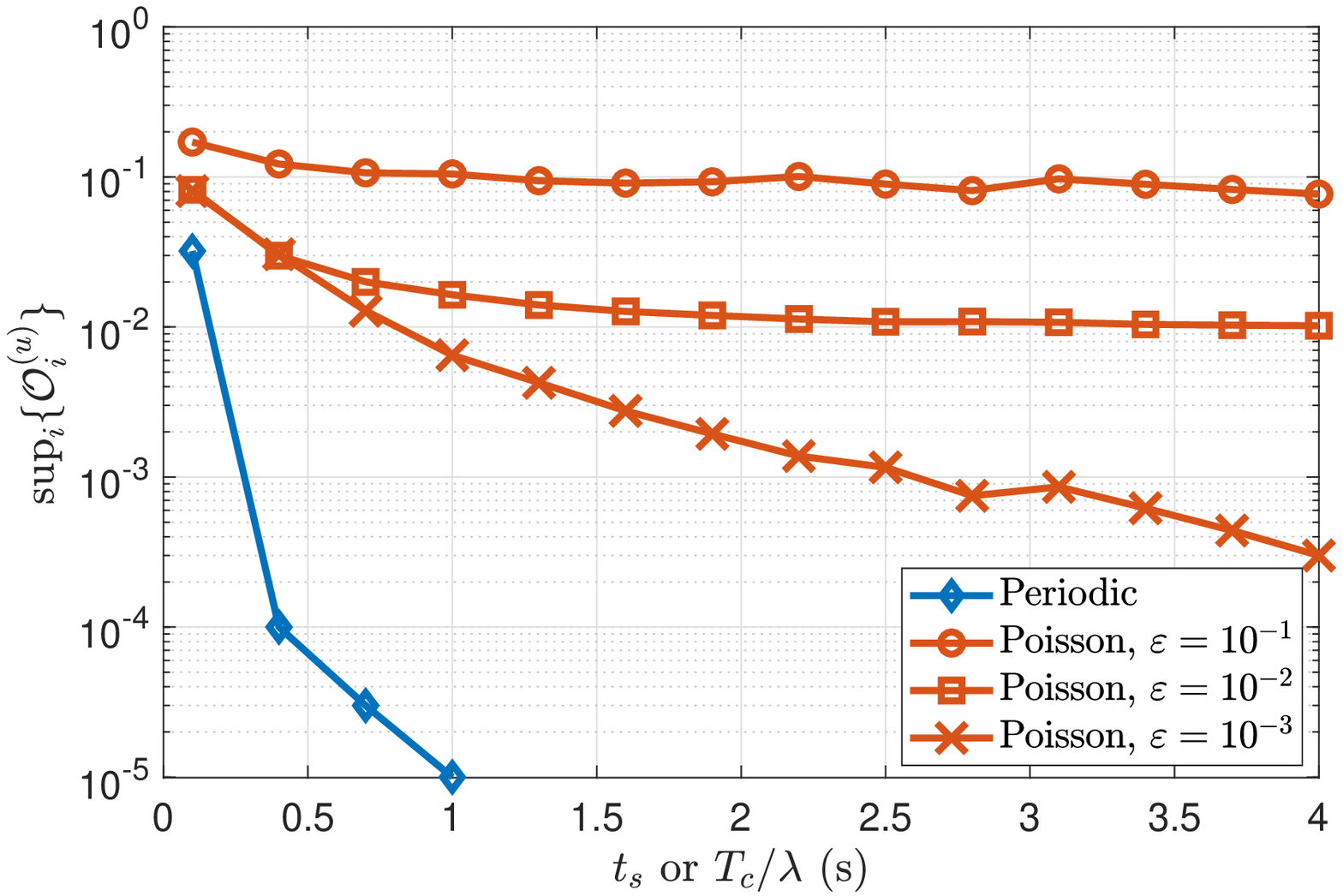}
		\caption{Worst-case information outage probability $a)$ Performance as a function of the number of devices in a Poisson traffic scenario (top). $b)$ Performance of MMSE as a function of the average number of coherence time intervals between the transmission of consecutive messages from each device (bottom). In both cases, we set $M_r=2$.}			
		\label{Fig5}
	\end{figure}
	\subsection{On the general performance}\label{generalP}
    \begin{figure}[t!]
    	\centering
    	\includegraphics[width=0.95\columnwidth]{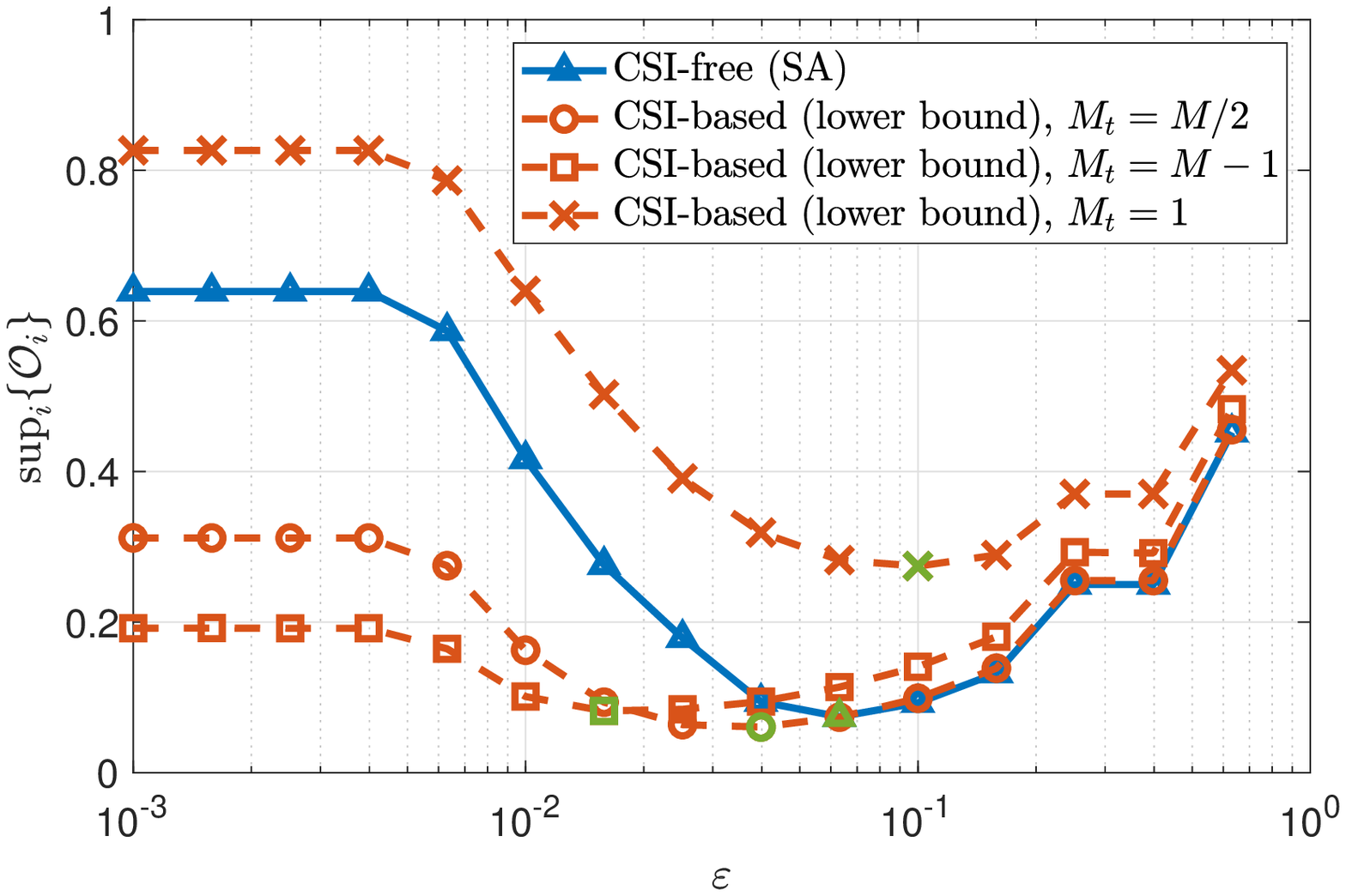}\vspace{3mm}\\
    	\includegraphics[width=0.95\columnwidth]{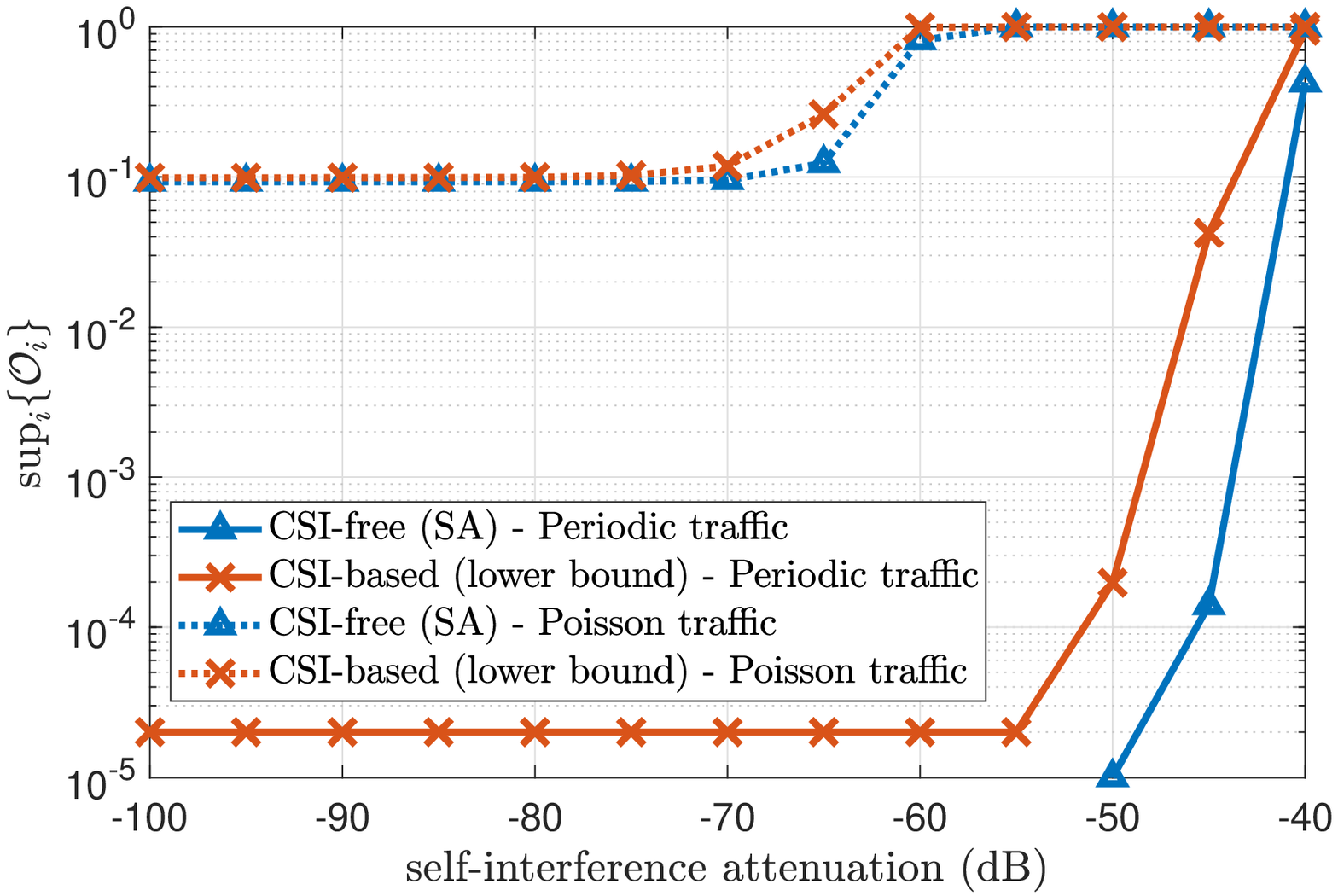}
    	\caption{Worst-case outage probability  $a)$ in a Poisson traffic scenario as a function of the target collision  probability (top) and $b)$ in both Poisson (with $\varepsilon=0.1$) and periodic traffic scenarios as a function of the self-interference attenuation in case of imperfect SIC and $M_t=M/2$ (bottom).}
    	\label{Fig6}
    \end{figure}
	Herein we investigate the overall outage performance by taking into consideration both the energy and information outage performances. As a first result and for a Poisson traffic scenario, we show in Fig.~\ref{Fig6}a the overall worst-case outage probability   as a function of the target collision  probability. In case of the CSI-based WET scheme, we utilize different antenna partitions and found out that the system benefits more by having no less transmit antennas than receive antennas at the HAP since the WET phase is the most critical. Besides, the main insight is that for each configuration there is an optimum target {collision  probability} as discussed in Remark~\ref{re8}, which at the end influences significantly the overall system performance. Such optimum values are highlighted in green in the figure. On the other hand,  we illustrate the impact of an imperfect SIC in the overall system performance in case of both periodic and Poisson traffic profiles in Fig.~\ref{Fig6}b. In case of Poisson traffic, we set $\varepsilon=0.1$ and consequently adopted $M_t=M/2$ which was shown to be a suitable choice in Fig.~\ref{Fig6}a. For simplicity, the self-interfering near-field channels are modeled as constant, and their associated path-loss along with the effective attenuation after SIC are considered in the self-interference attenuation parameter in the $x-$axis of Fig.~\ref{Fig6}b. Note that the system performance is degraded only if SIC performs extremely poorly. However, state-of-the-art SIC techniques can already attenuate the self-interfering signals to noise floor in practice \cite{Alves.2020}, which means that SIC operation is not critical in the considered setup and in the remaining results we keep our assumption of perfect SIC. Still, notice that the performance under Poisson traffic is sensitive to imperfect SIC because there are usually more concurrent transmissions for which the noise power plus self-interference is more harmful. Finally, observe that the system under periodic traffic behaves better, while the CSI-free SA scheme attains the best performance in all the cases shown with imperfect SIC. The latter means that the CSI acquisition costs overcame the beamforming gains.

	\begin{figure}[t!]
		\includegraphics[width=0.95\columnwidth]{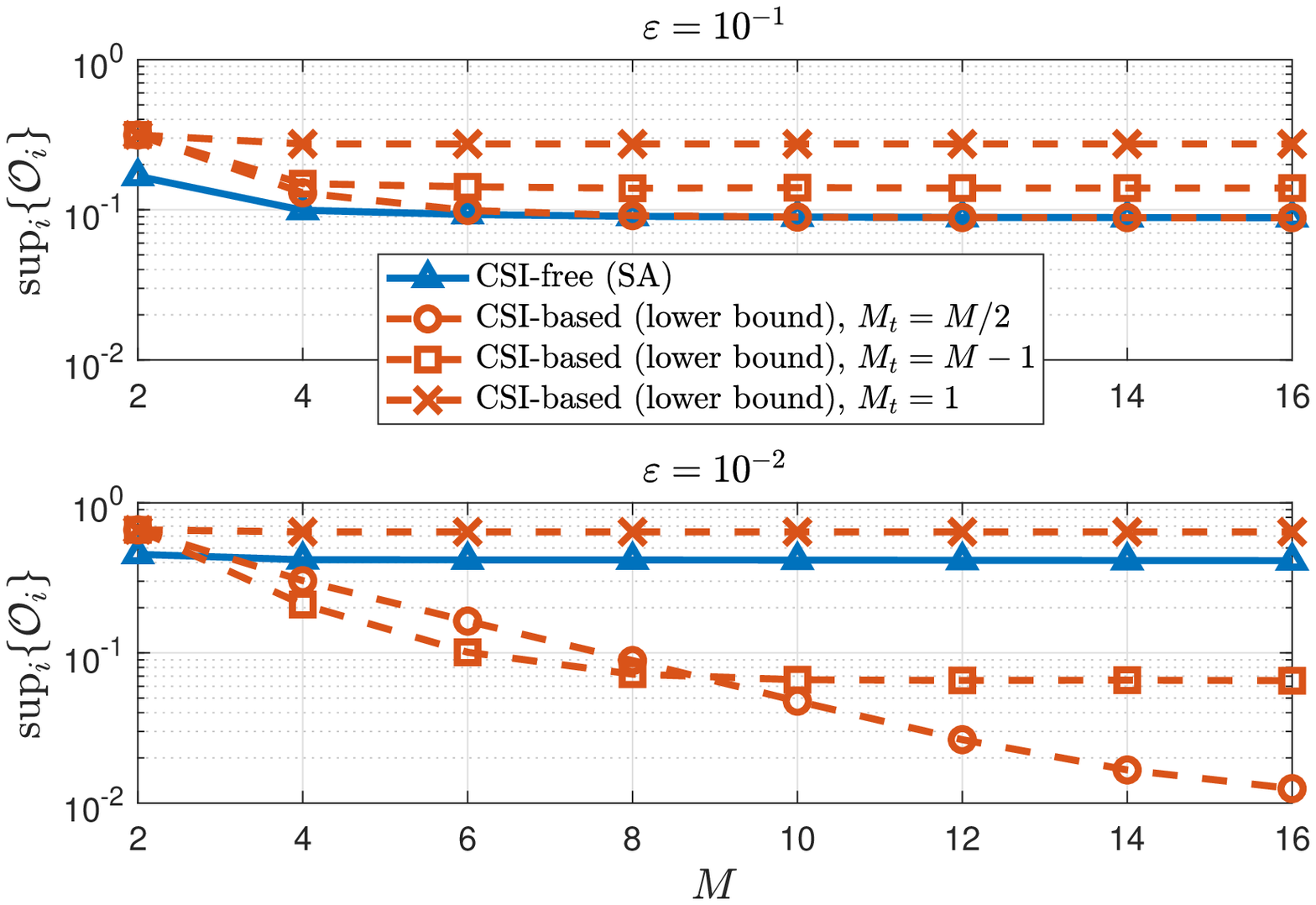}\vspace{3mm}\\
		\includegraphics[width=0.95\columnwidth]{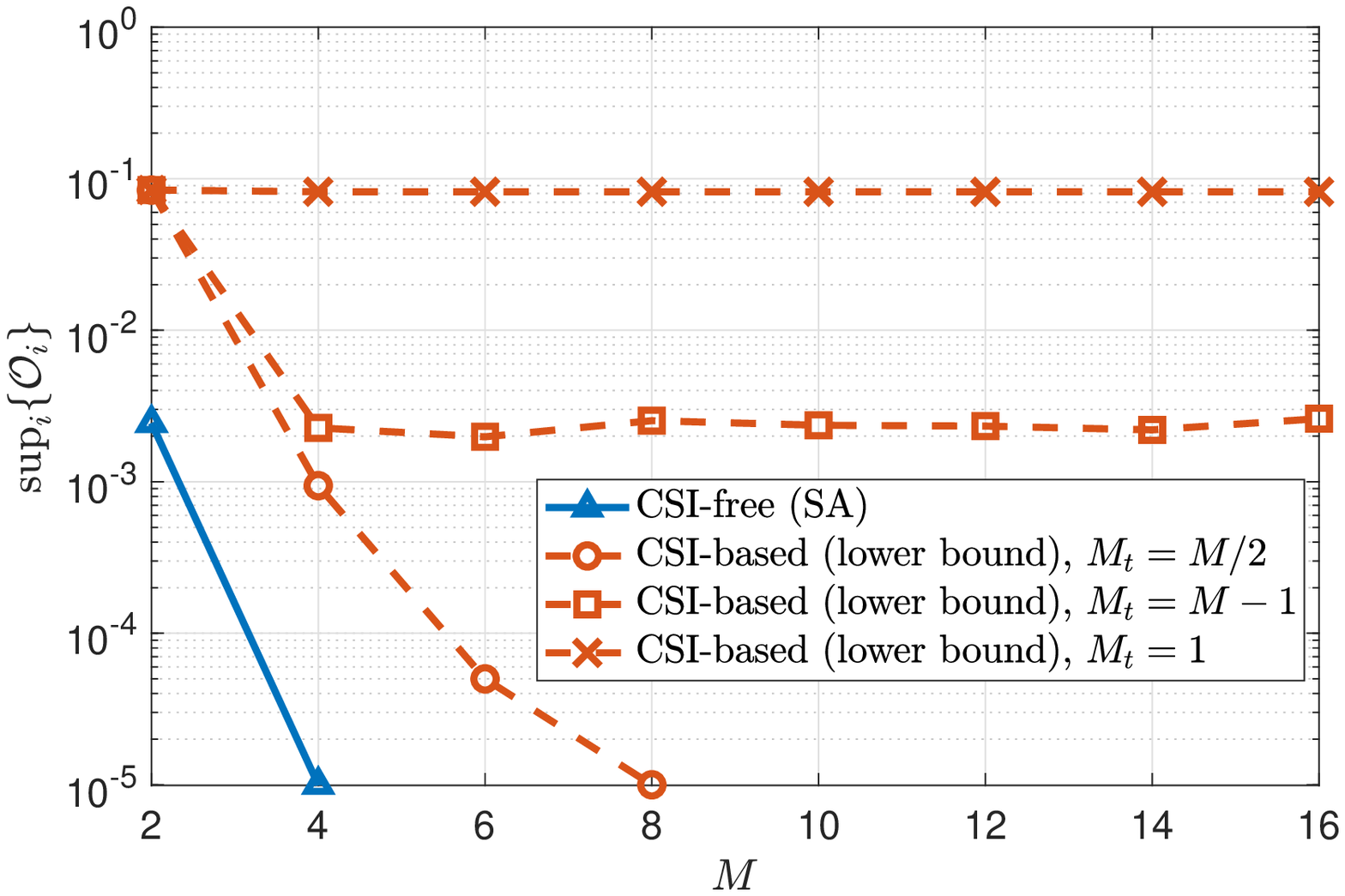}
		\caption{Worst-case outage probability as a function of the number of antennas. $a)$ Poisson traffic, $\varepsilon\in \{10^{-1}, 10^{-2}\}$ (top); and $b)$ Periodic traffic (bottom).}	
		\label{Fig7}
	\end{figure}
	Fig.~\ref{Fig7} shows the performance for both Poisson (Fig.~\ref{Fig7}a), with $\varepsilon\in\{10^{-1},10^{-2}\}$, and periodic (Fig.~\ref{Fig7}b) traffic as a function of the total number of antennas.  In case of Poisson traffic, the performance highly depends on $\varepsilon$ as commented in the previous paragraph, and corroborated now in Fig.~\ref{Fig7}a. One can see that the optimum $\varepsilon$ is around $10^{-1}$ for $M\lesssim 4$, while a more stringent value should be adopted for $M\gtrsim 4$. Also, observe that for relatively small $M$, the CSI-free SA scheme is preferable, while as $M$ increases, the CSI-based alternative becomes more suitable, but the set of antennas must be properly partitioned for transmitting and receiving. In case of periodic traffic, the performance improvements as a function of $M$ are even more noticeable. It is shown that, while having most of the antennas dedicated to transmission is acceptable for small $M$, there is need of a more equitable distribution of the transmit and receive antennas as $M$ increases. Anyway, the CSI-free SA scheme outperforms all the CSI-based configurations in the examples illustrated in Fig.~\ref{Fig7}b.
	
	\begin{figure}[t!]
		\includegraphics[width=0.97\columnwidth]{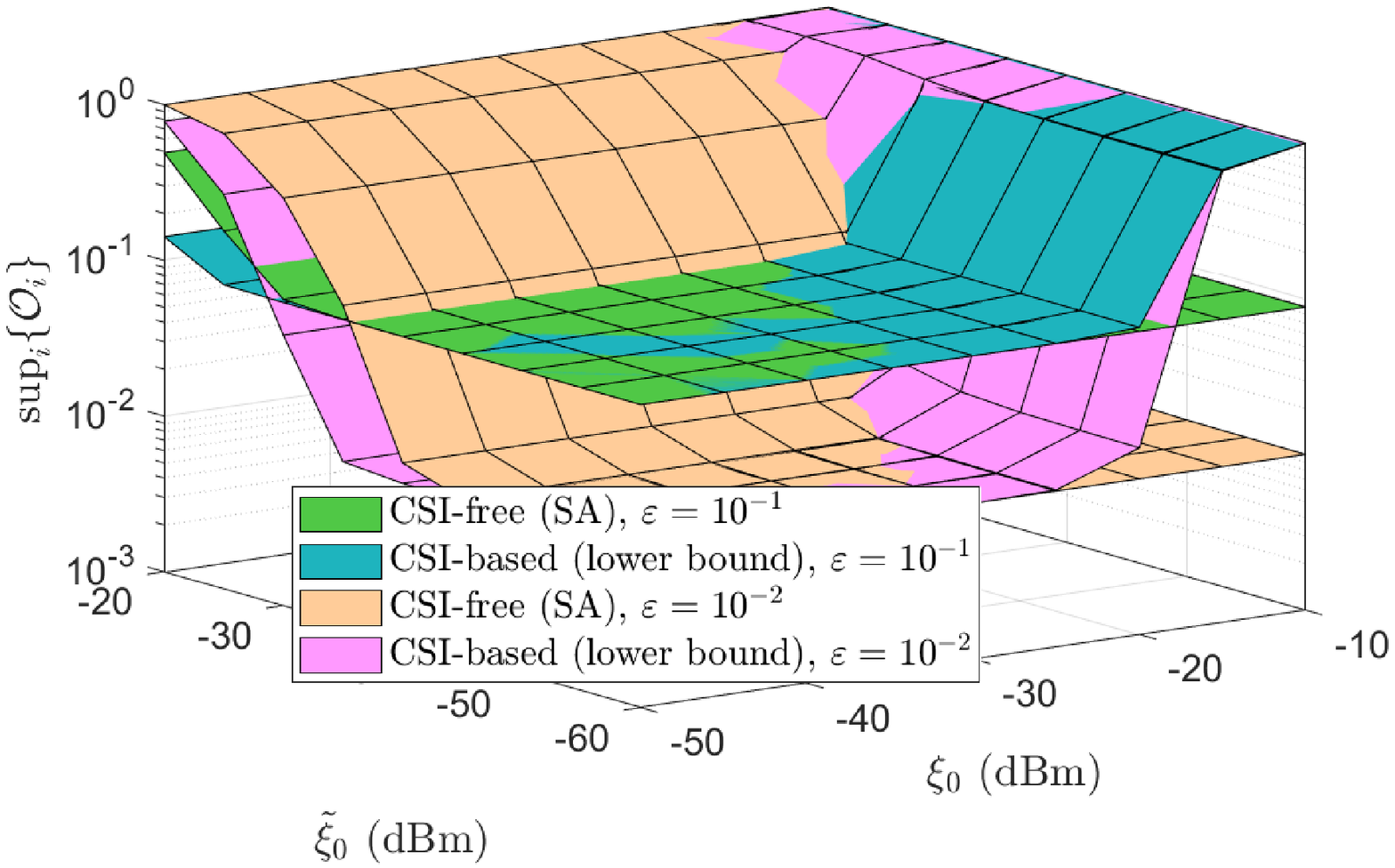}\vspace{3mm}\\
		\includegraphics[width=0.95\columnwidth]{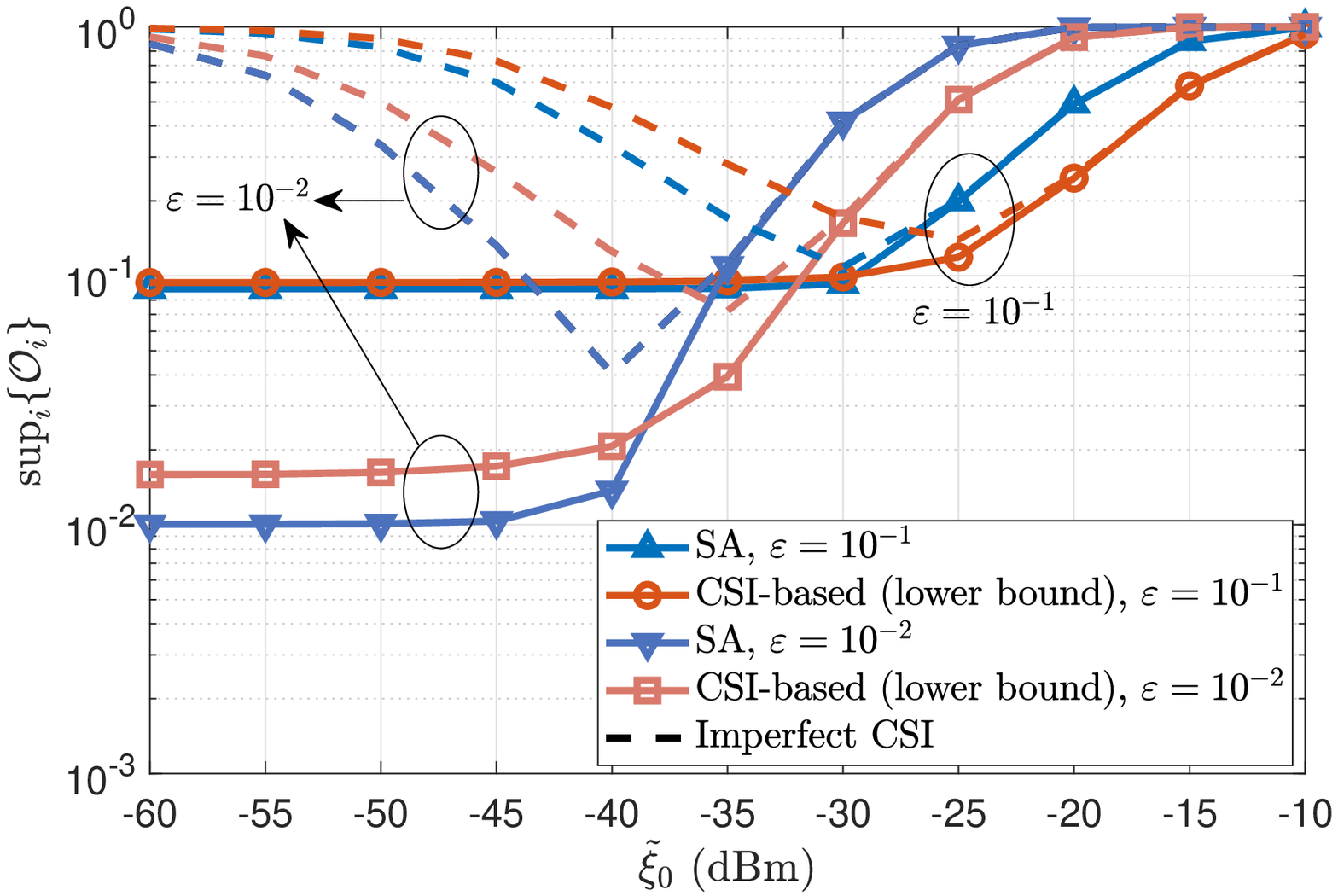}
		\caption{Worst-case outage probability under Poisson traffic with $\varepsilon\in\{10^{-1},10^{-2}\}$. $a)$ Performance as a function of $\xi_0$ and $\tilde{\xi}_0$ under perfect CSI acquisition (top). $b)$ Performance comparison of perfect and imperfect CSI acquisition for $\xi_0\!=\!-20$ dBm and $t_p=70\ \mu$s (bottom).  We set $M_t=M_r=M/2$ for the CSI-based schemes.}		
		\label{Fig8}
	\end{figure}
	The overall outage probability as a function of both downlink and uplink CSI nominal acquisition costs, $\xi_0$ and $\tilde{\xi}_0$, respectively, is shown in Fig.~\ref{Fig8}a. We focus on the Poisson traffic scenario with $\varepsilon\in\{10^{-1},10^{-2}\}$, for which the performance under SA and its CSI-based WET counterpart is more balanced\footnote{Although not shown here, SA outperforms more easily its CSI-based counterpart under periodic traffic, similar to results and discussions associated to previous figures.}. As observed, SA keeps a constant performance along the $x-$axis since the HAP does not require/use any CSI for powering the devices, while the performance under the CSI-based scheme is seriously affected as $\xi_0$ increases above $-20$ dBm. Meanwhile, the overall performance decreases as $\tilde{\xi}_0$ takes significant values since CSI is required for information decoding under both CSI-based and CSI-free WET mechanisms. As observed, the CSI-free WET scheme becomes attractive as $\tilde{\xi}_0$ decreases. Note that throughout this work we have assumed perfect CSI for information decoding. However, a performance degradation due to imperfect CSI acquisition is unavoidable in practice. In Fig.~\ref{Fig8}b, we show some preliminary hints on the expected performance under imperfect uplink CSI. The results are drawn after decomposing the uplink channel vector as $\mathbf{h}_i^{(u)}=\mathbf{\hat{h}}_i^{(u)}+\mathbf{\tilde{h}}_i^{(u)}$, where
		\begin{figure}[t!]
			\includegraphics[width=0.95\columnwidth]{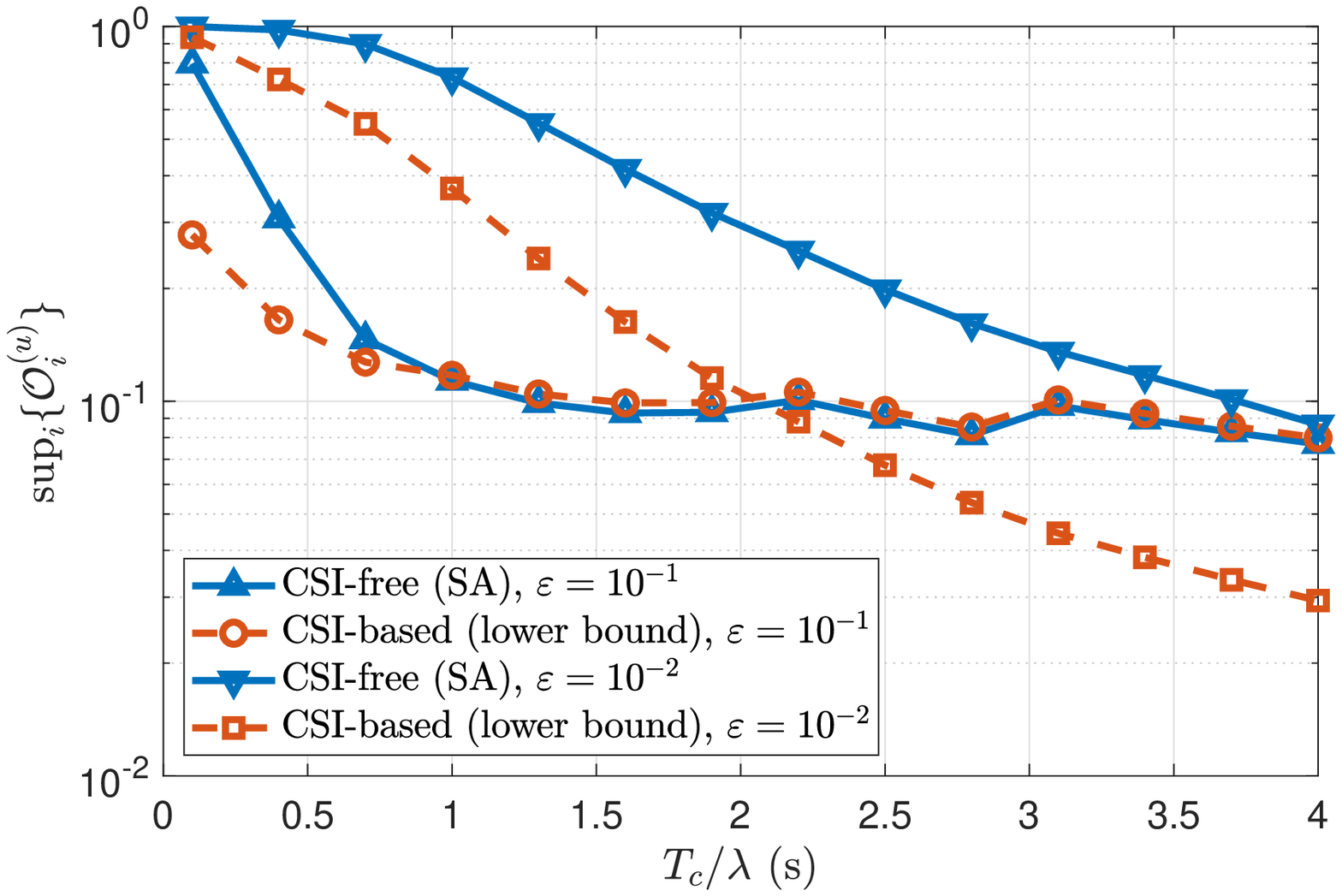}\vspace{3mm}\\ 
			\includegraphics[width=0.95\columnwidth]{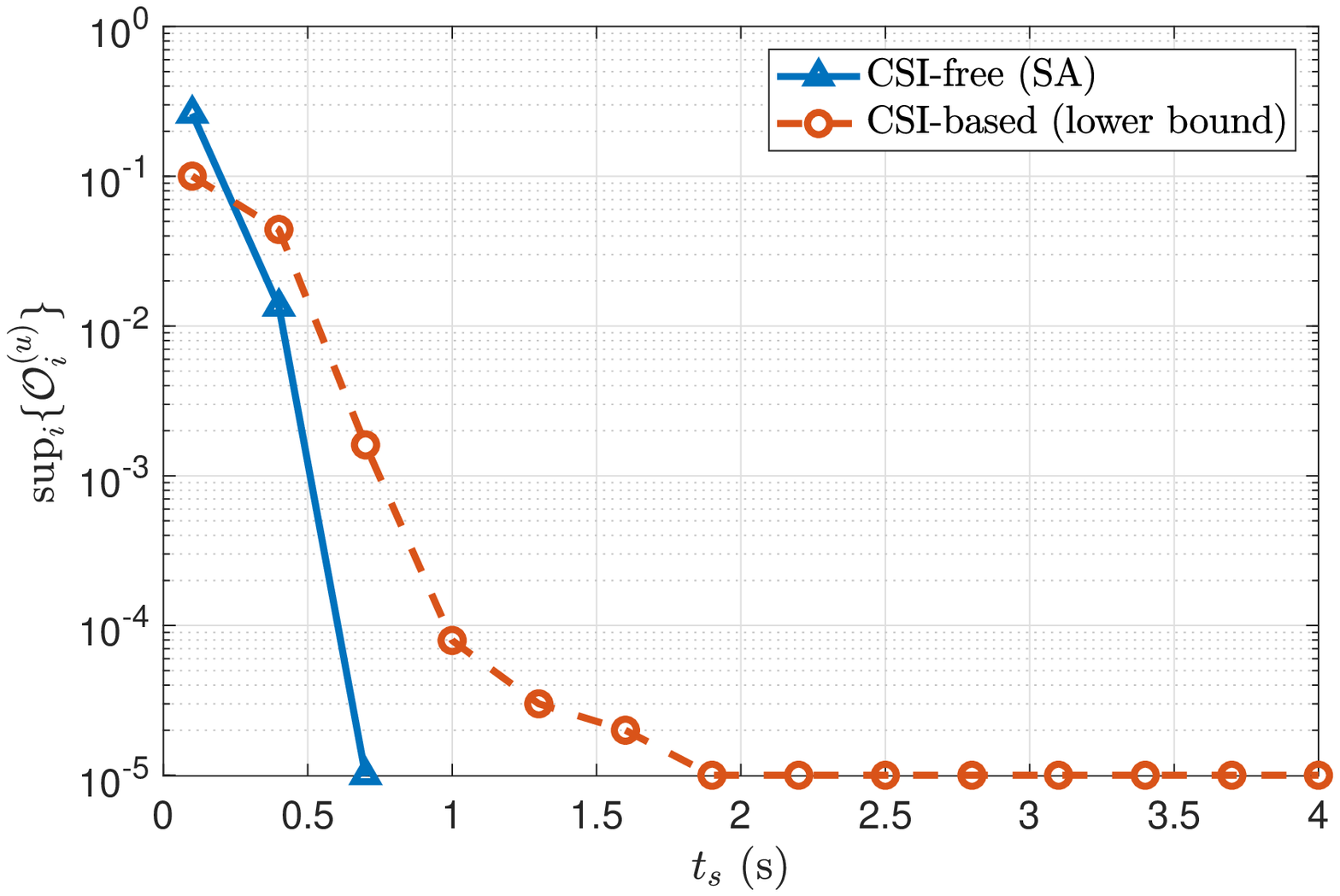}
			\caption{Worst-case outage probability as a function of $a)$ $T_c/\lambda$ for Poisson traffic (top),  $b)$ $t_s$ for periodic traffic (bottom). We set $M_t=M_r=M/2$ for the CSI-based schemes.}
			\label{Fig9}
		\end{figure}
		\begin{align}
		\mathbf{\hat{h}}_i^{(u)}\!&\sim \mathcal{CN}\bigg(\sqrt{\frac{\kappa}{1+\kappa}}\mathbf{1}_{M_t\times 1},\frac{1}{1+\kappa}\times\frac{\xi_{\mathrm{csi}}^{(u)}/t_p}{\xi_{\mathrm{csi}}^{(u)}/t_p\!+\!\sigma^2}\mathbf{I}_{M_t\times M_t}\bigg),\nonumber\\
		\mathbf{\tilde{h}}_i^{(u)}\!&\sim\mathcal{CN}\bigg(\mathbf{0},\frac{1}{1+\kappa}\times\frac{\sigma^2}{\xi_{\mathrm{csi}}^{(u)}/t_p+\sigma^2}\mathbf{I}_{M_t\times M_t}\bigg) \nonumber
		\end{align}
 are the channel estimate and its corresponding estimation error, respectively, and $t_p$ is the pilot symbol time \cite{LopezFernandez.2018}. The system performance degrades for both relatively small and large values of $\tilde{\xi}_0$, which affect the information and energy outage probabilities, respectively. This evidences a fundamental trade-off in setting $\tilde{\xi}_0$, that needs to be considered in practice, and which we plan to deepen in future works.

	Fig.~\ref{Fig9} shows the overall outage probability as a function of the average number of coherence time intervals between the transmission of consecutive messages from each device. As such average inter-arrival time increases, the chances of outage decrease. In case of Poisson traffic (Fig.~\ref{Fig9}a), the performance is strictly determined by the chosen target collision  probability (see Remark~\ref{re8}), whose optimum value tends to decrease as the average inter-arrival time increases. In case of periodic traffic (Fig.~\ref{Fig9}b), the overall performance improves, but bounded, with the inter-arrival time. The minimum outage probability is already reached for $t_s=2$ s, and this is due to the same arguments we exposed earlier when discussing Fig.~\ref{Fig5}b results. 
	\section{Conclusion}\label{conclusions}
	In this paper, we assessed for the first time the suitability of CSI-based and CSI-free multi-antenna WET schemes in a WPCN with a massive number of devices and under periodic or Poisson traffic sources. The system performance was evaluated, and optimized whenever possible, in terms of the worst (farthest) node's performance for the sake of fairness, and considering a realistic power consumption model at the devices. In case of the CSI-based WET scheme, we cast the optimization problem as an SDP problem, hence, a global solution is perfectly available by using regular optimization solvers. Additionally, we showed that a MRT beamformer is close to the optimum whenever the farthest node experiences at least 3 dB of power attenuation more than the remaining devices. As a CSI-free scheme, we adopted the novel $\mathrm{SA}$ strategy introduced and analyzed in \cite{Lopez.2019_CSI,Lopez.2020}. This scheme, although not capable of providing higher average harvesting gains compared to the CSI-based schemes, it does provide greater WET/WIT  diversity gain with lower energy requirements. Our numerical results evidenced that the CSI-free scheme performs specially favorably under periodic traffic conditions, while its performance may degrade significantly if the setup is not optimally configured in case of Poisson traffic. In fact, the system performance not only deteriorates under Poisson random access when compared to deterministic traffic, but  optimally configuring the network becomes also more challenging.
	In that regard, we cast an optimization problem to determine the optimum pilot reuse factor such that the collision probability under Poisson accesses remains below a certain limit. Numerical results demonstrated the existence of an optimum target collision  probability. 
	Finally, we showed the considerable gains from using a MMSE equalizer instead of a ZF equalizer in the analyzed WPCN scenario.
	\appendices
	\section{Proof of Theorem~\ref{the2}}\label{App_A}
	We depart from \eqref{O1} and use \eqref{E0} and \eqref{Ei''} to write
	\begin{align}
	\sup_i\{\mathcal{O}_i^{(d)}\}&\ge \mathbb{P}\big[E_{i'}<E_0\big]\nonumber\\
	&=\mathbb{P}\Big[\eta T_c\!\sum_V\! E_{i'}^\mathrm{rf}\!<\!v\xi_\mathrm{csi}^{(d)}\!+\!\xi_\mathrm{csi}^{(u)}\!+\!p_c vT_c\!+\!p t\Big]\nonumber\\
	&\stackrel{(a)}{=}\mathbb{E}_V\bigg[\mathbb{P}\Big[\frac{\eta T_c P\beta_{i'}}{2(1+\kappa)}\sum_V \mathfrak{X}<v\xi_\mathrm{csi}^{(d)}+\xi_\mathrm{csi}^{(u)}\nonumber\\
	&\qquad\qquad\qquad\qquad\qquad+p_c vT_c+p t\Big|V\Big]\bigg]\nonumber\\
	&\stackrel{(b)}{=}1-\mathbb{E}_V\bigg[\mathcal{Q}_{M_t v}\Big(\sqrt{2M_t\kappa v},\sqrt{\frac{2(1+\kappa)}{\eta T_c P\beta_{i'}}}\times\nonumber\\
	&\qquad\qquad\ \ \  \sqrt{v\big(\xi_\mathrm{csi}^{(d)}+p_cT_c\big)\!+\!\xi_\mathrm{csi}^{(u)}\!+\!p t}\Big)\bigg],
	\end{align}
	where $(a)$ comes from averaging the outage events conditioned on a given $v$ and  using \eqref{Ei'}, while $(b)$ follows by taking the summation of $v$ non-central chi-squared RVs, which obeys a non-central chi-squared distribution as well, but with $v$ times the number of degrees of freedom and non-centrality parameter, and using its CDF by taking advantage of \eqref{cdf}.
	Then, \eqref{cfO} is attained  after taking the expectation with respect to $V$ by using \eqref{Vpdf}. We avoided using an infinite notation here, and instead considered only the first $v_{\max}$ summands, hence \eqref{cfO} is, in general, an approximation that becomes exact as $v_{\max}\rightarrow\infty$. However,
	since $V$ is a discrete exponential-like random variable characterized in \eqref{Vpdf}, setting $v_{\max}$ such that $v_{\max}\ge 10\times \mathbb{E}[V]$ is enough for a good accuracy. Notice that
	\begin{align}
	\mathbb{E}[V]=\frac{e^\lambda}{e^\lambda-1},\label{Ev}
	\end{align}
	which follows from realizing that computing $\mathbb{E}[V]=\sum_{v=1}^{\infty}ve^{-\lambda v}$ is equivalent to evaluate  $e^{\lambda}$ into the $Z-$transform of the sequence $1,2,3,\cdots$, which is $\frac{z^{-1}}{(1-z^{-1})^2}$. \hfill 	\qedsymbol
	\section{Proof of Theorem~\ref{the3}}\label{App_B}
    Let us focus on the performance in terms of average incident RF power in a certain device $s_i\in\mathcal{S}\backslash s_{i'}$. By using \eqref{Ei} and the second MRT beamfomer given in \eqref{wj}, we have that
	\begin{align}
	\mathbb{E}\big[E_i^\mathrm{rf}\big]&=\mathbb{E}\bigg[P\beta_i\sum_{j=1}^{M_t}\big|(\mathbf{h}_i^{(d)})^T\mathbf{w}_j\big|^2\bigg]\nonumber\\
	&=P\beta_i\mathbb{E}\Big[\big|(\mathbf{h}_i^{(d)})^T\mathbf{w}_1\big|^2\Big]\nonumber\\
	&=P\beta_i\mathbb{E}\bigg[\frac{\big|\mathbf{h}_{i'}^{(d)H}\mathbf{h}_i^{(d)}\big|^2}{||\mathbf{h}_{i'}^{(d)}||^2}\bigg].\label{exp}
	\end{align}
	Unfortunately, $\mathbf{w}_1$ follows a cumbersome projected normal distribution when $\Im\{\mathbf{w}_1\}=\mathbf{0}$ \cite{Hernandez.2017}, which makes the analysis of the distribution of $\big|(\mathbf{h}_i^{(d)})^T\mathbf{w}_1\big|^2$ already very complicated even for such a simplified scenario. Meanwhile, decoupling the expression as shown in the last line of \eqref{exp} does not solve the problem since numerator and denominator are correlated. 
	We resorted to simulation and standard fitting procedures, and found out that
	\begin{align}
	\frac{\mathbb{E}\big[E_i^\mathrm{rf}\big]}{P\beta_i}&\approx\frac{1}{4}\Big(\frac{\kappa}{1+\kappa/\sqrt{2}}\Big)^2M_t+\frac{1}{1+\kappa/2}\label{app}
	\end{align}
	matches \eqref{exp} accurately, which is corroborated in Fig.~\ref{Fig11}. Now, based on \eqref{Ei'} we have that 
	\begin{figure}[t!]
		\centering 
		\includegraphics[width=0.95\columnwidth]{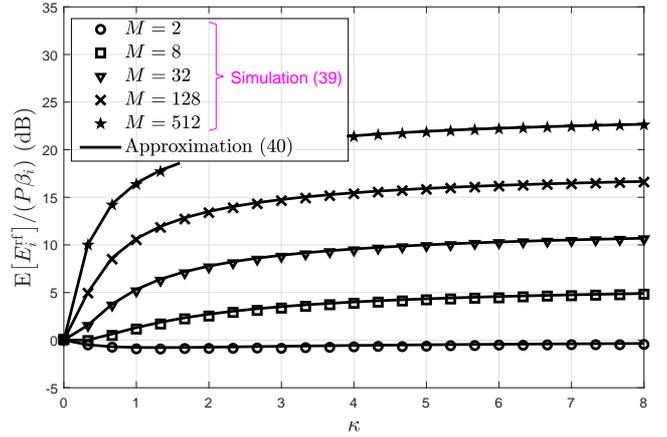}
		\caption{$\mathbb{E}\big[E_i^\mathrm{rf}\big]/(P\beta_i)$ vs $\kappa$ for $M_t\in\{2,8,32,128,512\}$. Comparison between the Monte Carlo-based \eqref{exp} and the analytical approximation \eqref{app}.} 
		\label{Fig11}
	\end{figure}
	\begin{align}
	\mathbb{E}\big[E_{i'}^\mathrm{rf}\big]=P\beta_{i'}M_t,
	\end{align}
	and define $\Omega=\inf_{s_i\in\mathcal{S}\backslash s_{i'}}\big\{\mathbb{E}\big[ E_i^\mathrm{rf}\big]\big\}\big/\mathbb{E}\big[E_{i'}^\mathrm{rf}\big]$, which matches \eqref{Ec}. By using such definition, it is evident that as $\Omega$ grows, the MRT beamformer affects less the non-intended receivers and consequently becomes more frequently the optimum. 
	In fact, already when $\Omega> 1$, we have that even when the HAP uses only the CSI statistics referred to $s_{i'}$, the remaining devices harvest more energy at least half of the time. The reason is that the median of the distribution of $E_i^\mathrm{rf}$ is smaller than $E_{i'}^\mathrm{rf}$'s for the same average performance, i.e., $\Omega=1$, since the distribution of $E_i^\mathrm{rf}$ is intrinsically more positively skewed than $E_{i'}^\mathrm{rf}$'s. Therefore, under such circumstances, the MRT beamformer is at least half of the time the optimum from a system perspective.   \hfill 	\qedsymbol

\section{Proof of Theorem~\ref{the4}}\label{App_C}
Let us denote as $\hat{\mathcal{S}}\subseteq\tilde{\mathcal{S}}$, where $N'=|\hat{\mathcal{S}}|$, the set of devices using the same pilot signal, then
\begin{align}
\mathbb{P}[s_i\in\hat{\mathcal{S}}]=\frac{1}{L}\mathbb{P}[s_i\in\tilde{\mathcal{S}}]=\frac{t}{LT_c}\big(1-e^{-\lambda}\big),\label{ps2}
\end{align}
and similar to $N$, $N'$ is a Binomial RV with parameters $\mathrm{S}$ and $\frac{t}{LT_c}\big(1-e^{-\lambda}\big)$. Consequently,  $\mathbb{E}[N']\!=\frac{\mathrm{S}t}{LT_c}\big(1-e^{-\lambda}\big)$ represents the average number of concurrent transmissions of devices using the same pilot signals.
Here we focus our attention to the performance of $s_{i'}$. 
Assuming such a device is already active, its associated collision probability is then given by 
\begin{align}
\mathcal{O}_\mathrm{col}&=1-\mathbb{P}[N'=1|N'>0]=1-\frac{\mathbb{P}[N'=1]}{1-\mathbb{P}[N'=0]},
\end{align}
which matches \eqref{ocol}. \hfill	\qedsymbol
\section{Algorithm for solving problem in \eqref{problem}}\label{App_D}
	Let us take 
	\begin{align}
	u=1-\frac{t}{LT_c}(1-e^{-\lambda}),\label{a0}
	\end{align}
	then, by substituting \eqref{a0} into \eqref{ocol}, the problem in \eqref{problem} can be easily addressed after solving for $u$ the following inequality
	\begin{align}
\mathcal{O}_\mathrm{col}\le \varepsilon\ \ \ \rightarrow\ \ \ 
1-\frac{\mathrm{S}(1-u)u^{\mathrm{S}-1}}{1-u^\mathrm{S}}&\le \varepsilon\nonumber\\
\frac{(1-u)u^{\mathrm{S}-1}}{1-u^\mathrm{S}}&\ge \frac{1-\varepsilon}{\mathrm{S}}.\label{a1}
\end{align}
	Notice that the left term is an increasing function of $u$ in the interval $(0,1)$, which is the interval of interest. In fact, for $u\rightarrow \{0,\ 1\}$ the left term converges to $\{0,\ 1\}$, respectively, and since $\frac{1-\varepsilon}{\mathrm{S}}\in (0,1)$, a unique solution is guaranteed.
	Now, by relaxing the inequality to an equality and making $g(u)=\frac{1-u}{1-u^\mathrm{S}}$, we reformulate \eqref{a1} as
\begin{align}
u^{\mathrm{S}-1}&=\frac{1-\varepsilon}{\mathrm{S}g(u)}\nonumber\\
u&=\Big(\frac{1-\varepsilon}{\mathrm{S}}\Big)^{\frac{1}{\mathrm{S}-1}}g(u)^{-\frac{1}{\mathrm{S}-1}}.\label{gu}
\end{align}
	\begin{algorithm}[t!]
	\SetKwInOut{Input}{Input}
	\SetKwInOut{Output}{Output}
	\Input{$\mathrm{S}$, $\varepsilon$, $t$, $T_c$, $\lambda$ and tolerance $u_\epsilon>0$}
	\Output{$L_0$, $\mathrm{iter}$}
	\textbf{Initializing:}\label{lin0} $u^{(0)}=\big(\frac{\mathrm{S}}{1-\varepsilon}-1\big)^{-\frac{1}{\mathrm{S}-1}}$, $\mathrm{iter}=1$, $\Delta u=\infty$\\
	\While{$\Delta u>u_\epsilon$}
	{
		$u^{(\mathrm{iter})}:=\tilde{g}\big(u^{(\mathrm{iter}-1)}\big)$ as given in \eqref{gu2}\;	
		$\Delta u=\big|u^{(\mathrm{iter})}-u^{(\mathrm{iter}-1)}\big|$\;
		$\mathrm{iter}:=\mathrm{iter}+1$\;
	}
	{
		$L_0=\bigg\lceil \frac{t\big(1-e^{-\lambda}\big)}{\big(1-u^{(\mathrm{iter})}\big)T_c} \bigg\rceil$.\label{linF}
	}   
	\caption{Finding $L_0$ \eqref{problem}}
\end{algorithm}
	Thus, we can say that the unique solution of \eqref{gu}, $u^*$, is a fixed point of
	\begin{align}
	\tilde{g}(u)=\Big(\frac{1-\varepsilon}{\mathrm{S}}\Big)^{\frac{1}{\mathrm{S}-1}}g(u)^{-\frac{1}{\mathrm{S}-1}}.\label{gu2}
	\end{align}
	Now, notice that
	\begin{align}
	|\tilde{g}(u)'|&=\frac{1}{\mathrm{S}-1}\Big(\frac{1-\varepsilon}{\mathrm{S}}\Big)^{\frac{1}{\mathrm{S}-1}}g(u)^{-1-\frac{1}{\mathrm{S}-1}}|g'(u)|\nonumber\\
	&=\!\frac{1}{\mathrm{S}\!-\!1}\!\Big(\!\frac{1\!-\!\varepsilon}{\mathrm{S}}\!\Big)^{\frac{1}{\mathrm{S}\!-\!1}}\!\!g(u)^{-\!1\!-\!\frac{1}{\mathrm{S}\!-\!1}}\Bigg|\frac{\mathrm{S}(1\!-\!u)u^{\mathrm{S}\!-\!1}\!\!-\!u(1\!-\!u^\mathrm{S})}{(1-u^\mathrm{S})^2}\Bigg|\nonumber\\
	&=\!\frac{1}{\mathrm{S}\!-\!1}\!\Big(\frac{1\!-\!\varepsilon}{\mathrm{S}}\Big)^{\frac{1}{\mathrm{S}\!-\!1}}\!\!g(u)^{1\!-\!\frac{1}{\mathrm{S}\!-\!1}}\!\bigg|\frac{\mathrm{S}u^{\mathrm{S}\!-\!1}}{1-u}\!-\!\frac{u(1\!-\!u^{\mathrm{S}})}{(1-u)^2}\bigg|,\label{a3}
	\end{align}
	which reaches the maximum for $\varepsilon=0$. Fig.~\ref{Fig12}a shows \eqref{a3} for such extreme configuration, and since $|\tilde{g}(u)'| < 1$, we can assure that that still holds for any $\varepsilon$. Then, based on the Fixed Point Theory \cite{Agarwal.2001} the convergence to the solution is guaranteed by using a fixed point iterative procedure as the one presented in Algorithm~1. Notice that  one can choose any $u^{(0)}\in(0,1)$ as initial value, however, we chose the value shown in line~\ref{lin0} as it already constitutes a good guess towards the final value $u^*$, which helps to reduce the required number of iterations. Such an initial value comes from realizing that $\frac{1-u^\mathrm{S}}{1-u}=\sum_{n=0}^{\mathrm{S}-1}u^n>1+u^{\mathrm{S}-1}$ (using the geometric series) and substituting such result into \eqref{a1} to attain
	%
\begin{align}
\frac{u^{\mathrm{S}-1}}{1+u^{\mathrm{S}-1}}=\frac{1-\varepsilon}{\mathrm{S}}\ \ \ \rightarrow\ \ \
u=\Big(\frac{\mathrm{S}}{1-\varepsilon}-1\Big)^{-\frac{1}{\mathrm{S}-1}}.
\end{align}
	For Algorithm~1 to run, we require to specify a tolerance error $u_\epsilon$ that we are willing to accept, and it constitutes the stopping criterion. The smaller  $u_\epsilon$ is, more iterations are required as corroborated in Fig.~\ref{Fig12}b, where it can also be observed that less than 16 iterations are enough in all the cases. Another interesting fact is that the convergence is even faster as $\mathrm{S}$ increases. After convergence, $L_0$ is computed according to line~\ref{linF}, which comes from isolating $L$ in \eqref{a0}. \hfill 	\qedsymbol
	
	\begin{figure}[t!]
		\includegraphics[width=0.95\columnwidth]{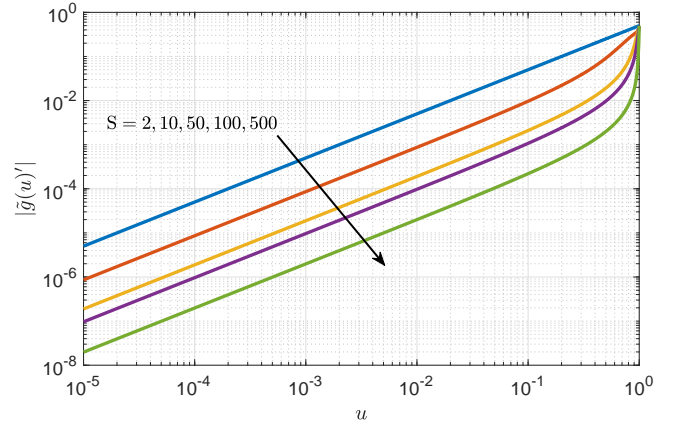}\vspace{4mm}\\		
		\vspace{0mm}\ \  \includegraphics[width=0.93\columnwidth]{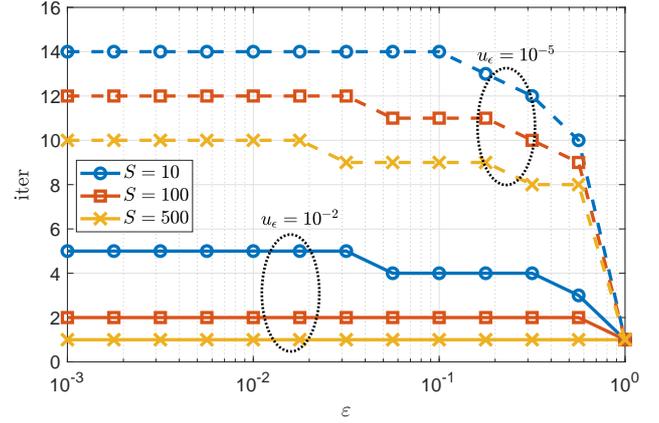}
		\caption{$a)$ $|\tilde{g}(u)'|$ vs $u$ for $\varepsilon=0$ (top). $b)$ Required number of iterations as a function of $\varepsilon$ for Algorithm~1 to converge. We use $\mathrm{S\in\{10,100,500\}}$ and $u_\epsilon\in\{10^{-2},10^{-5}\}$ (bottom).}
		\label{Fig12}
	\end{figure}

	\bibliographystyle{IEEEtran}
	\bibliography{IEEEabrv,references}

\begin{thebibliography}{10}
\providecommand{\url}[1]{#1}
\csname url@samestyle\endcsname
\providecommand{\newblock}{\relax}
\providecommand{\bibinfo}[2]{#2}
\providecommand{\BIBentrySTDinterwordspacing}{\spaceskip=0pt\relax}
\providecommand{\BIBentryALTinterwordstretchfactor}{4}
\providecommand{\BIBentryALTinterwordspacing}{\spaceskip=\fontdimen2\font plus
\BIBentryALTinterwordstretchfactor\fontdimen3\font minus
  \fontdimen4\font\relax}
\providecommand{\BIBforeignlanguage}[2]{{%
\expandafter\ifx\csname l@#1\endcsname\relax
\typeout{** WARNING: IEEEtran.bst: No hyphenation pattern has been}%
\typeout{** loaded for the language `#1'. Using the pattern for}%
\typeout{** the default language instead.}%
\else
\language=\csname l@#1\endcsname
\fi
#2}}
\providecommand{\BIBdecl}{\relax}
\BIBdecl

\bibitem{Matti.2019}
M.~Latva-aho and K.~Lepp{\"a}nen, ``Key drivers and research challenges for
  {6G} ubiquitous wireless intelligence,'' \emph{6G Research Visions}, vol.~1,
  2019.

\bibitem{Mahmood.2020}
\BIBentryALTinterwordspacing
N.~H. {Mahmood \textit{et al.}}, ``White paper on critical and massive machine
  type communication towards {6G} [white paper],'' \emph{6G Research Visions},
  vol.~11, 2020. [Online]. Available:
  \url{http://urn.fi/urn:isbn:9789526226781}
\BIBentrySTDinterwordspacing

\bibitem{sullivan.2020}
\BIBentryALTinterwordspacing
``{Frost \& Sullivan Visionary Innovation Group: Mega Trends}.'' [Online].
  Available: \url{https://ww2.frost.com/research/visionary-innovation/}
\BIBentrySTDinterwordspacing

\bibitem{Portilla.2019}
J.~{Portilla \textit{et al.}}, ``The extreme edge at the bottom of the
  {Internet of Things}: A review,'' \emph{IEEE Sensors J.}, vol.~19, no.~9, pp.
  3179--3190, May 2019.

\bibitem{Niyato.2017}
D.~{Niyato \textit{et al.}}, ``Wireless powered communication networks:
  Research directions and technological approaches,'' \emph{IEEE Wirel.
  Commun.}, vol.~24, no.~6, pp. 88--97, Dec 2017.

\bibitem{Lopez.2019}
O.~L.~A. {L{\'o}pez \textit{et al.}}, ``Massive wireless energy transfer:
  Enabling sustainable {IoT} towards {6G} era,'' \emph{arXiv preprint
  arXiv:1912.05322}, 2019.

\bibitem{Ghazanfari.2016}
A.~{Ghazanfari \textit{et al.}}, ``Ambient {RF} energy harvesting in
  ultra-dense small cell networks: performance and trade-offs,'' \emph{IEEE
  Wirel. Commun.}, vol.~23, no.~2, pp. 38--45, 2016.

\bibitem{Clerckx.2019}
B.~{Clerckx \textit{et al.}}, ``Fundamentals of wireless information and power
  transfer: From {RF} energy harvester models to signal and system designs,''
  \emph{IEEE J. Sel. Areas Commun.}, vol.~37, no.~1, pp. 4--33, Jan 2019.

\bibitem{Mahmood.2019}
N.~H. {Mahmood \textit{et al.}}, ``Six key features of machine type
  communication in {6G},'' in \emph{6G Summit}, 2020, pp. 1--5.

\bibitem{Ju.2014}
H.~{Ju} and R.~{Zhang}, ``Throughput maximization in wireless powered
  communication networks,'' \emph{IEEE Trans. Wireless Commun.}, vol.~13,
  no.~1, pp. 418--428, January 2014.

\bibitem{Lopez.2017}
O.~L.~A. {L\'opez \textit{et al.}}, ``Ultrareliable short-packet communications
  with wireless energy transfer,'' \emph{IEEE Signal Process. Lett.}, vol.~24,
  no.~4, pp. 387--391, April 2017.

\bibitem{Huang.2016}
W.~{Huang \textit{et al.}}, ``On the performance of multi-antenna
  wireless-powered communications with energy beamforming,'' \emph{IEEE Trans.
  Veh. Technol.}, vol.~65, no.~3, pp. 1801--1808, March 2016.

\bibitem{Chen.2015}
H.~{Chen \textit{et al.}}, ``Harvest-then-cooperate: Wireless-powered
  cooperative communications,'' \emph{IEEE Trans. Signal Process.}, vol.~63,
  no.~7, pp. 1700--1711, April 2015.

\bibitem{Lopez.2018}
O.~L.~A. {L\'opez \textit{et al.}}, ``Wireless powered communications with
  finite battery and finite blocklength,'' \emph{IEEE Trans. Commun.}, vol.~66,
  no.~4, pp. 1803--1816, April 2018.

\bibitem{LopezFernandez.2018}
O.~L. {A. L\'opez \textit{et al.}}, ``Ultra-reliable cooperative short-packet
  communications with wireless energy transfer,'' \emph{IEEE Sensors J.},
  vol.~18, no.~5, pp. 2161--2177, March 2018.

\bibitem{Makki.2016}
B.~{Makki \textit{et al.}}, ``Wireless energy and information transmission
  using feedback: Infinite and finite block-length analysis,'' \emph{IEEE
  Trans. Commun.}, vol.~64, no.~12, pp. 5304--5318, Dec 2016.

\bibitem{Cantos.2019}
L.~{Cantos} and Y.~H. {Kim}, ``Max-min fair energy beamforming for wireless
  powered communication with non-linear energy harvesting,'' \emph{IEEE
  Access}, vol.~7, pp. 69\,516--69\,523, 2019.

\bibitem{Du.2020}
R.~{Du \textit{et al.}}, ``Wirelessly-powered sensor networks: Power allocation
  for channel estimation and energy beamforming,'' \emph{IEEE Trans. Wireless
  Commun.}, pp. 1--1, 2020.

\bibitem{LopezMahmood.2020}
O.~L.~A. {{L\'opez}, \textit{et al.}}, ``Ultra-low latency, low energy, and
  massiveness in the {6G} era via efficient {CSIT}-limited scheme,'' \emph{IEEE
  Commun. Mag.}, vol.~58, no.~11, pp. 56--61, 2020.

\bibitem{Clerckx.2018}
B.~{Clerckx} and J.~{Kim}, ``On the beneficial roles of fading and transmit
  diversity in wireless power transfer with nonlinear energy harvesting,''
  \emph{IEEE Trans. Wireless Commun.}, vol.~17, no.~11, pp. 7731--7743, Nov
  2018.

\bibitem{Lopez.2019_CSI}
O.~L.~A. {L\'opez \textit{et al.}}, ``Statistical analysis of multiple antenna
  strategies for wireless energy transfer,'' \emph{IEEE Trans. Commun.},
  vol.~67, no.~10, pp. 7245--7262, Oct 2019.

\bibitem{Lopez.2020}
O.~L.~A. {L\'opez \textit{et al..}}, ``On {CSI}-free multi-antenna schemes for
  massive {RF} wireless energy transfer,'' \emph{IEEE Internet Things J.}, pp.
  1--1, 2020.

\bibitem{Kobayashi.2011}
H.~{Kobayashi \textit{et al.}}, \emph{Probability, random processes, and
  statistical analysis: applications to communications, signal processing,
  queueing theory and mathematical finance}.\hskip 1em plus 0.5em minus
  0.4em\relax Cambridge University Press, 2011.

\bibitem{Nuttall.1975}
A.~{Nuttall}, ``Some integrals involving the {$Q_{\!M}\!$} function
  (corresp.),'' \emph{IEEE Trans.\! Inf.\! Theory}, vol.~21, no.~1, pp. 95--96,
  1975.

\bibitem{Khan.2016}
T.~A. {Khan \textit{et al.}}, ``On wirelessly powered communications with short
  packets,'' in \emph{IEEE GC Wkshps}, Dec 2016, pp. 1--6.

\bibitem{Guillaud.2005}
M.~{Guillaud \textit{et al.}}, ``A practical method for wireless channel
  reciprocity exploitation through relative calibration.'' in \emph{ISSPA},
  2005, pp. 403--406.

\bibitem{Proakis.2001}
J.~Proakis, ``Digital communications,'' 2001.

\bibitem{Nikaein.2013}
N.~{Nikaein \textit{et al.}}, ``Simple traffic modeling framework for machine
  type communication,'' in \emph{ISWCS}, Aug 2013, pp. 1--5.

\bibitem{Thudugalage.2016}
A.~{Thudugalage \textit{et al.}}, ``Beamformer design for wireless energy
  transfer with fairness,'' in \emph{IEEE ICC}, 2016, pp. 1--6.

\bibitem{Ye.2011}
Y.~Ye, \emph{Interior point algorithms: theory and analysis}.\hskip 1em plus
  0.5em minus 0.4em\relax John Wiley \& Sons, 2011, vol.~44.

\bibitem{Shao.2020}
X.~{Shao \textit{et al.}}, ``A dimension reduction-based joint activity
  detection and channel estimation algorithm for massive access,'' \emph{IEEE
  Trans. Signal Process.}, vol.~68, pp. 420--435, 2020.

\bibitem{Siriteanu.2012}
C.~{Siriteanu \textit{et al.}}, ``{MIMO Zero-Forcing} detection analysis for
  correlated and estimated {Rician} fading,'' \emph{IEEE Trans. Veh. Technol.},
  vol.~61, no.~7, pp. 3087--3099, Sep. 2012.

\bibitem{Gao.1998}
H.~{Gao \textit{et al.}}, ``Theoretical reliability of {MMSE} linear diversity
  combining in {Rayleigh}-fading additive interference channels,'' \emph{IEEE
  Trans. Commun.}, vol.~46, no.~5, pp. 666--672, May 1998.

\bibitem{Lim.2019}
H.~{Lim} and D.~{Yoon}, ``On the distribution of {SINR for MMSE MIMO}
  systems,'' \emph{IEEE Trans. Commun.}, vol.~67, no.~6, pp. 4035--4046, June
  2019.

\bibitem{Rubinstein.2016}
R.~Y. Rubinstein and D.~P. Kroese, \emph{Simulation and the {Monte Carlo}
  method}.\hskip 1em plus 0.5em minus 0.4em\relax John Wiley \& Sons, 2016,
  vol.~10.

\bibitem{Alves.2020}
H.~{Alves, \textit{et al.}}, \emph{Full-Duplex Communications for Future
  Wireless Networks}.\hskip 1em plus 0.5em minus 0.4em\relax Springer, 2020.

\bibitem{Hernandez.2017}
D.~{Hernandez-Stumpfhauser \textit{et al.}}, ``The general projected normal
  distribution of arbitrary dimension: modeling and {Bayesian} inference,''
  \emph{Bayesian Analysis}, vol.~12, no.~1, pp. 113--133, 2017.

\bibitem{Agarwal.2001}
R.~{Agarwal \textit{et al.}}, \emph{Fixed point theory and applications}.\hskip
  1em plus 0.5em minus 0.4em\relax \!\!\!\!Cambridge university press, 2001,
  vol. 141.

\end{thebibliography}
\end{document}